\documentclass[journal]{IEEEtran}

\input xy
\xyoption{all}

\usepackage{amsfonts}
\usepackage{amsmath}
\usepackage{amssymb}
\usepackage{amscd}
\usepackage{amsthm}
\usepackage{algorithm}
\usepackage{algorithmic}
\usepackage{latexsym} 
\usepackage{mathrsfs}
\usepackage{graphicx}
\usepackage{indentfirst}
\usepackage{color}
\usepackage{booktabs}
\usepackage{graphicx}
\usepackage{url}
\usepackage{slashbox}

\newtheorem{thm}{Theorem} 
\newtheorem{lem}[thm]{Lemma}
\newtheorem{co}[thm]{Corollary} 
\newtheorem{prop}[thm]{Proposition}
\newtheorem{rem}[thm]{Remark}
\theoremstyle{definition}
\newtheorem{ex}[thm]{Example}
\newtheorem{defi}[thm]{Definition}

\newcommand{\F}{\mathbb{F}}

\newcommand{\N}{\mathbb{N}}

\newcommand{\Z}{\mathbb{Z}}

\newcommand{\PP}{\mathcal{P}}

\newcommand{\Gr}{\mathcal{G}_{q}(k,n)}

\newcommand{\Cvs}{\mathcal{C}}
\newcommand{\Svs}{\mathcal{S}}
\newcommand{\Evs}{\mathcal{E}}
\newcommand{\Rvs}{\mathcal{R}}

\newcommand{\Vvs}{\mathcal{V}}
\newcommand{\Uvs}{\mathcal{U}}
\newcommand{\PS}{\mathcal{P}_q(n)}

\newcommand{\mat}[1]{\left(\begin{matrix}#1\end{matrix} \right)}
\newcommand{\rs}{\mathrm{rs}}

\newcommand{\rank}{\mathrm{rank}}

\newcommand{\orderspaces}[3]{\genfrac{[}{]}{0pt}{}{#1}{#2}_{#3}}

\newcommand{\diag}{\mathrm{diag}}
\newcommand{\ord}{\mathrm{ord}}

\newcommand{\wt}{\mathrm{weight}}
\newcommand{\stab}[2]{\Stab_{#1}({#2})}

\DeclareMathOperator{\Stab}{Stab}
\DeclareMathOperator{\RCF}{RCF}

\DeclareMathOperator{\Aut}{Aut}
\DeclareMathOperator{\Anti}{Anti}
\DeclareMathOperator{\Sym}{Sym}
\DeclareMathOperator{\GammaL}{\Gamma L}
\DeclareMathOperator{\GL}{GL}
\DeclareMathOperator{\dec}{dec}
\DeclareMathOperator{\Iso}{Iso}
\DeclareMathOperator{\AntiIso}{AntiIso}

\begin{document}

\title{Cyclic Orbit Codes}


  
  
  

\author{Anna-Lena Trautmann, Felice Manganiello, Michael Braun and
  Joachim Rosenthal 
  \thanks{A.-L. Trautmann and J. Rosenthal are with the University of
    Zurich. The authors were partially supported by
    Swiss National Science Foundation Project no.~126948.}  \thanks{F. Manganiello is with the University of
    Toronto. He was partially supported by Swiss National Science
    Foundation Projects no.~126948 and no.~135934.}  \thanks{M. Braun
    is with the University of Applied Sciences Darmstadt.} }
      
\maketitle

\begin{abstract}
In network coding a constant dimension code consists of a set of $k$-dimensional subspaces of $\F_q^n$. 
Orbit codes are constant dimension codes which are defined as orbits of a subgroup of the general linear group, acting on the set of all subspaces of $\F_q^n$. If the acting group is cyclic, the corresponding orbit codes are called cyclic orbit codes. In this paper we give a classification of cyclic orbit codes and propose a decoding procedure for a particular subclass of cyclic orbit codes.
\end{abstract}


%
\section{Introduction and related work}
\label{sec:intro}

\PARstart{N}{etwork} coding describes a method for attaining a maximum information flow within a network, that is an acyclic directed graph with possibly several sources and sinks. The notion first arose in \cite{ah00}. An algebraic approach to coding in non-coherent networks, called \emph{random network coding}, with respect to error correction and the corresponding transmission model is given in the paper of K\"otter and Kschischang \cite{ko08}. There they specify under which conditions errors can be recognized and corrected.

In the random network coding setting a message corresponds to a subspace $\mathcal{V}$ of a finite $n$-dimensional vector space over a finite field ${\mathbb F}_q$, denoted by $\mathcal{V} \le {\mathbb F}_q^n$. Generating vectors of this subspace will be injected at the sources of the network and will be transmitted to their neighboring nodes. These nodes gather some vectors, linearly combine them, and transmit these linear combinations to their neighbors. In the end the receiver nodes collect linear combinations of the original injected vectors. 

In the error-free case the receiver can reconstruct the whole subspace $\Svs$ from the collected vectors. 

In real world applications, networks are exposed to noise such that messages can be lost or modified during the transmission of $\mathcal{V}$. It is also possible that some node wants to disrupt the transmission flow by injecting wrong vectors. Thus, on one hand, some vectors of $\mathcal{V}$ might be lost and a smaller subspace $\mathcal{V}'<\mathcal{V}$ will be received. On the other hand, vectors which are not contained in $\mathcal{V}$ might be received. These erroneous vectors span a vector space $\Evs$, thus $\Rvs= \mathcal{V}'\oplus \Evs$ will be received.  Vectors lost during transmission are called \emph{erasures} and additional received vectors not contained in $\mathcal{V}$ are called \emph{errors}.

Since K\"otter and Kschischang published their pioneering article \cite{ko08} several papers about the structure of network codes appeared, e.\,g. \cite{BEV11,el10,et08u,EV08,et08p,ko08p,ma08p,ma11,si08a,sk10,tr10p,tr11}. A survey on all most important results including constructions and bounds of subspace codes which appeared before 2010 offers the paper of Khaleghi et al. \cite{KSK09}.

In \cite{ko08p} Kohnert and Kurz constructed constant dimension codes by solving a Diophantine system of linear equations. In order to obtain a reduction of the huge system of equations they prescribed a group of automorphisms which is a subgroup of the general linear group. The corresponding codes are unions of orbits of that prescribed group on $k$-subspaces. This approach by prescribing a group of automorphisms and solving a Diophantine system of equations was first introduced by Kramer and Mesner \cite{KM76} in 1976 for the construction of combinatorial $t$-designs.

The concept of defining codes as orbits of certain groups traces back to Slepian \cite{Sle68} where  Euclidian spaces were considered---the corresponding codes are called \emph{group codes}. In the network coding setting Trautmann et al. \cite{tr10p} defined such codes as \emph{orbit codes}. The applied groups are again subgroups of the general linear group. If the acting group is the Singer cycle group Elsenhans et al. \cite{el10} gave a decoding algorithm for constant 3-dimensional orbit codes. A subclass of orbit codes define the cyclic subgroups of the general linear groups, the corresponding orbit codes are called \emph{cyclic orbit codes}---they were introduced by Trautmann and Rosenthal in \cite{tr11}, where the authors started with a partial characterization of these cyclic orbit codes.

The goal of this paper is to continue the work of \cite{tr11}. We describe cyclic orbit codes, give a complete characterization of these codes and propose a decoding algorithm for irreducible cyclic orbit codes which is different from the approach introduced in \cite{el10}.

The paper is organized as follows: In the second section we give the major preliminaries and consider some different point of views on subspace codes.

In the third section we motivate and describe orbit codes in a general setting by considering group actions on metric sets. Furthermore, we give an interpretation of maximum likelihood decoding in terms of the group actions.

In Section \ref{sec:classification} we classify cyclic orbit codes. For each type given by the conjugacy class of cyclic subgroups of the general linear group we characterize the corresponding cyclic orbit codes. Section \ref{sec:cardinality} deals with their size and the minimum distance.

In Section \ref{sec:decoding} we investigate a decoding algorithm for cyclic orbit codes arising from Singer cycles. We compare the complexity of this algorithm with different decoding procedures already proposed for other types of constant dimension codes.

Finally we conclude with the summary of the main results of this paper, give some suggestions for further research, and gather some open questions.

\section{Background}
\label{sec:background}

Let $\F_q$ be the finite field with $q$ elements (where $q$ is a prime power) and let $n$ be natural number. 

When $\Uvs$ is a subspace of the vector space $\F_q^n$ we will use the notation $\mathcal{U}\le \F_q^n$. The dimension will be abbreviated by $\dim (\mathcal{U})$. The set of all subspaces of $\F_q^n$ of dimension $k$ is called \emph{Grassmann variety} or simply \emph{Grassmannian} and is denoted by
\[
\Gr := \{ \mathcal{U}\le \F_q^n\mid \dim(\mathcal{U})=k\}.
\] 
The cardinality of this set is given by the \emph{$q$-Binomial coefficient}, also called \emph{Gaussian number}, 
\[
|\Gr| = \orderspaces{n}{k}{q}:=\genfrac{}{}{}{}{(q^n-1)(q^{n-1}-1)\cdots(q^{n-k+1}-1)}{(q^k-1)(q^{k-1}-1)\cdots(q-1)}.
\]
The union of all Grassmann varieties, i.\,e. the set of all subspaces of $\F_q^n$,  is called the \emph{projective space} and denoted by 
\[
\PS = \bigcup_{k=0}^n \Gr .
\] 

\subsection{The lattice point of view}

The projective space $\PS$ forms a lattice with supremum \lq\lq$+$\rq\rq{} and infimum \lq\lq$\cap$\rq\rq{} which is called the \emph{linear lattice} \cite{Bir67}. The dimension of subspaces defines a rank function satisfying the Jordan-Dedekind chain condition, i.\,e. $\PS$ is a JD-lattice. It is well-known that each JD-lattice $(L,\le,\wedge,\vee,r,0,1)$ with partial order \lq\lq$\le$\rq\rq{}, infimum \lq\lq$\wedge$\rq\rq{}, supremum \lq\lq$\vee$\rq\rq{}, rank function $r$, zero-element $0=\wedge_{x\in L} x$, and one-element $1=\vee_{x\in L}$ defines a metric lattice with a distance function $d$ (see \cite{Bir67}):
\[
d(x,y):=r(x\vee y)-r(x\wedge y).
\]
In the projective space setting, the corresponding metric is called the \emph{subspace distance}: For two subspaces $\mathcal{U},\mathcal{V} \in \PS$ we obtain:
\begin{align*}
d(\mathcal{U},\mathcal{V})&= \dim(\mathcal{U}+\mathcal{V})-\dim(\mathcal{U}\cap \mathcal{V})\\
&=\dim(\mathcal{U})+\dim(\mathcal{V})-2\dim (\mathcal{U} \cap \mathcal{V}).
\end{align*}

In \cite{ko08} K\"otter and Kschischang suggested to use this metric subspace lattice for the purpose of coding in erroneous network communication. They defined a \emph{subspace code} or \emph{random network code} as a subset $\mathcal{C}$ of the projective space:
\[
\mathcal{C}\subseteq \PS.
\]
A codeword of such a code corresponds to a subspace ${\mathcal U}\in{\mathcal C}$.  

The minimum distance $d(\mathcal{C})$ of a subspace code $\mathcal{C}$ is defined in the usual way --- as the smallest distance between any two elements of $\mathcal{C}$.

A special class of subspace codes are \emph{constant dimension codes} \cite{ko08}. These are simply subsets of the Grassmannian. If a constant dimension code $\mathcal{C}\subseteq \Gr$ has minimum distance $d(\mathcal{C})$, we call $\mathcal{C}$ an $[n,d(\mathcal{C}),|\mathcal{C}|,k]$-code.

A subspace code can be considered the $q$-analog of a binary block code \cite{Bra11,BEV11} as follows: The set $P(n)$ of all subsets of the $n$-set $N:=\{0,\ldots,n-1\}$ forms a JD-lattice with partial order \lq\lq$\le$\rq\rq{}, intersection of sets \lq\lq$\cap$\rq\rq{} as infimum operator, union \lq\lq$\cup$\rq\rq{} as supremum, the order $|U|$ of a subset $U\in P(n)$ as rank function, the zero-element $\emptyset$, and the one-element $N$:
\[
(P(n),\subseteq,\cap,\cup,|\cdot|,\emptyset,N).
\]
Hence, the order of the symmetric difference of sets $U,V\in P(n)$,
\[
d(U,V)=|U\cup V| - |U\cap V| ,
\]
defines a metric on $P(n)$. Since each subset $U\in P(n)$ can be represented by its characteristic vector $u=(u_0,\ldots,u_{n-1})\in \F_2^n$ ($u_i=1$ if $i\in U$, otherwise $u_i=0$), a binary block code corresponds to a subset 
\[
{\mathcal C}\subseteq P(n).
\]
The given distance function $d$ is then nothing but the Hamming metric. A binary constant weight code is a subset ${\mathcal C}\subseteq P(n)$ whose elements $U\in {\mathcal C}$ have a fixed order $|U|=k$. 

If we replace sets by subspaces and their orders by dimensions, we get the $q$-analog situation: Instead of the power set lattice $(P(n),\cap,\cup,|\cdot|,\emptyset,N)$ we use the projective space lattice
\[
(\PS,\le,\cap,+,\dim,\{0\},\F_q^n).
\]
Then the $q$-analog of a binary block code is a subspace code and the $q$-analog of a binary constant weight code is a constant dimension code.

If $(L,\le,\wedge,\vee,r,0,1)$ denotes a metric lattice, a bijective mapping $f:L\to L$ is called \emph{isometric} on $L$ if and only if 
\[
d(f(x),f(y))=d(x,y)
\]
for all $x,y\in L$. 

From lattice theory we know the following result, that order-preserving and order-reversing bijections on lattices define isometries \cite{Bir67,BEV11}:

\begin{lem}\label{lem:isometries}
Let $(L,\le,\wedge,\vee,r,0,1)$ denote a metric lattice and let $f:L\to L$ denote a bijective mapping. Then the following equivalencies hold:
\begin{enumerate}
\item $f$ is order-preserving, i.\,e. $x\le y \Leftrightarrow f(x)  \le f(y)$ for all $x,y\in L$, if and only if $r(f(x))=r(x)$ for all $x\in L$ and $f$ is isometric on $L$.
\item $f$ is order-reversing, i.\,e. $x\le y \Leftrightarrow f(y)  \le f(x)$ for all $x,y\in L$, if and only if $r(f(x))=r(1)-r(x)$ for all $x\in L$ and $f$ is isometric on $L$.
\end{enumerate}
\end{lem}

If $\Aut(L)$ denotes the set of all bijective order-preserving mappings on $L$ (also called \emph{lattice automorphisms}) and $\Anti(L)$ denotes the set of all bijective order-reversing mappings on $L$ (also called \emph{lattice anti-automorphisms}) the following lemma holds:

\begin{lem}\label{lem:aut:anti}
Let $(L,\le,\wedge,\vee)$ be a lattice. Assume that there exists a mapping $f\in \Anti(L)$. Then 
\[
\Anti(L)=\Aut(L)\circ f :=\{g\circ f\mid g\in \Aut(L)\}.
\]
\end{lem}

\begin{proof}
We give a more general proof for partially ordered sets: Let $(A,\le)$ respectively $(B,\preceq)$ two finite partially ordered sets (posets). Then we define the set of poset isomorphisms, respectively the set of poset anti-automorphisms between $A$ and $B$:
\begin{align*}
\Iso(A,B):=\{&f:A\to B\mid f\text{ bijection},\\
& f(x)\preceq f(y)\Leftrightarrow x\le y\,\, \forall x,y\in A\}
\end{align*}
resp.
\begin{align*}
\AntiIso(A,B):=\{&f:A\to B\mid f\text{ bijection},\\
& f(x)\preceq f(y)\Leftrightarrow y\le x\,\, \forall x,y\in A\}
\end{align*}
Let $(A,\le)$ and $(B,\preceq)$ two posets. Assume that there exists a mapping $d\in \AntiIso(A,B)$ Then we show that 
\[
\AntiIso(A,A)=\Iso(B,A) \circ d
\]
\begin{description}
\item[\fbox {$\subseteq$}]
Let $f\in \AntiIso(L,L)$. Then $f^{-1}\in \AntiIso(L,L)$ and we get the following equivalence, if we apply $d$:
\begin{align*}
x\le yÊ&\iff f^{-1}(y)\le f^{-1}(x)\\
&\iff d(f^{-1}(x))\preceq d(f^{-1}(y))\\
&\iff (d\circ f^{-1})(x)\preceq (d\circ f^{-1})(y)
\end{align*}
which means that $d\circ f^{-1}\in \Iso(A,B)$. Hence there exists a mapping $g\in \Iso(A,B)$ with $g=d\circ f^{-1}$. Applying $g^{-1}\in \Iso(B,A)$ from the left and $f$ from the right to this identity yields $id= g^{-1}\circ d\circ f^{-1}$ respectively $f= g^{-1}\circ d\in \Iso(B,A)\circ d$.

\item[\fbox {$\supseteq$}]
Let $f\in \Iso(B,A)$. Then we have 
\begin{align*}
x\le y &\iff d(y)\preceq f(x)\\
&\iff f(d(y)) \le f(d(x))\\
&\iff (f\circ d)(y)\le (f\circ d)(x)
\end{align*}
i.\,e. $f\circ d\in \AntiIso(A,A)$.
\end{description}
\end{proof}

In case of the power set lattice, the set of all order-preserving bijective mappings on $P(n)$ are exactly the permutations on the $n$-set,
\[
\Aut(P(n))=\Sym(n):=\{\pi:N\to N\mid \pi\text{ bijection}\} .
\]
In addition, an order-reversing mapping on $P(n)$ is defined by the setwise complement $\overline{U}=N \setminus U$, i.\,e. it satisfies
\[
U\subseteq V \iff \overline{V}\subseteq \overline{U}.
\]
In the setting of binary block codes, $Aut(P(n))$ is the set of the \emph{permutation isometries}, or \emph{linear isometries}, and the setwise complement corresponds to bit-flipping, i.\,e. a $1$ becomes a $0$ and vice versa. Denote by $\overline{\mathcal{C}}$ the code obtained by flipping the bits of a binary code $\mathcal{C}$. Then $\overline{\mathcal{C}}$ and $\Cvs$ have the same minimum distance $d(\mathcal{C})=d(\overline{\mathcal{C}})$, since the Hamming distance remains the same under bit-flipping, i.\,e. $d(\overline{u},\overline{v})=d(u,v)$.

In the case of the projective space lattice $\PS$ we know from the Fundamental Theorem of Projective Geometry (see \cite{Art57,Bae52}) that the set of order-preserving bijections on $\PS$ corresponds exactly to the set of \emph{semi-linear} mappings
\[
\Aut(\PS)=\GammaL_n(\F_q):=\GL_n(\F_q)\rtimes \Aut(\F_q)
\]
which is the semi-direct product of the \emph{general linear} group $\GL_n(\F_q)$ and the Galois group $\Aut(\F_q)$. 

If two codes are bijective with respect to a semi-linear mapping $f\in \GammaL_n(\F_q)$ we call them \emph{semi-linearily isometric}. If there exists $f\in \GL_n(\F_q)\le \GammaL_n(\F_q)$, such that $f$ maps one code to the other, the codes are called \emph{linearily isometric}.

Furthermore, the \emph{orthogonal complement} defines an order-reversing bijection on $\PS$,
\[
\mathcal{U}\le \mathcal{V}\iff \mathcal{V}^\perp\le \mathcal{U}^\perp
\]
and hence the \emph{dual code} $\mathcal{C}^\perp:=\{\mathcal{U}^\perp\mid \mathcal{U}\in\mathcal{C}\}$ defines a code with the same minimum distance as $\mathcal{C}\subseteq\PS$. 

In particular, if $\mathcal{C}$ is a constant dimension $[n,d(\mathcal{C}),|\mathcal{C}|,k]$-code, the dual code $\mathcal{C}^\perp$ is an $[n,d(\mathcal{C}),|\mathcal{C}|,n-k]$-code. In fact, the dual of a constant dimension code is the $q$-analog of a flipped binary constant weight code.

\subsection{The incidence geometry point of view}

Random network codes can be regarded as incidence structures in the projective space, therefore they are related to designs over finite fields \cite{Tho87}. A $t-(n,k,\lambda;q)$-design is a set $\mathcal{C}$ of $k$-subspaces, $\mathcal{C}\subseteq \Gr$, such that each $t$-subspace is contained in exactly $\lambda$ elements of $\mathcal{C}$. For the number of elements in such a design we get:
\[
|\mathcal{C}|=\lambda\genfrac{}{}{}{}{\orderspaces{n}{t}{q}}{\orderspaces{k}{t}{q}}.
\]
Each $t-(n,k,1;q)$-design defines an optimal constant dimension code (cf. \cite{ko08p}) with parameters
\[
[n,2(k-t+1),\genfrac{}{}{}{}{\orderspaces{n}{t}{q}}{\orderspaces{k}{t}{q}},k].
\] 
So far the only non-trivial $t-(n,k,1;q)$-designs are known for $t=1$. Such a $1$-design is called a \emph{spread}, which exists if and only if $k$ divides $n$ (cf. \cite{hi98,LW01}). The corresponding optimal constant dimension code with parameters $[n,2k,\genfrac{}{}{}{}{q^n-1}{q^k-1},k]$ is called a \emph{spread code} \cite{ma08p}.

\section{The group action point of view: orbit codes}
\label{sec:orbitcodes}

In this section we consider a third point of view: we derive random network codes from group actions---such codes are called \emph{orbit codes} \cite{ma11,tr10p}. We show, that this group theoretic approach yields a certain generalization of an ordinary linear code to random network codes. Both classes of codes, ordinary linear codes $C\le \F_q^n$ and orbit codes $\mathcal{C}\subseteq \PS$ can be defined by an orbit of a group action. If the group acting on the subspaces is cyclic we speak of \emph{cyclic orbit codes} \cite{tr11}. 

First we introduce the basic definitions and notation of groups acting on sets we need for our further investigations. The definitions and most results can be found in \cite{DM96,Ker99}.

\subsection{Basic definitions and properties of group actions}

Let $G$ be a finite multiplicative group with one-element $1_G$ and let $X$ denote a finite set. A \emph{group action} of $G$ on $X$ is a mapping 
\begin{align*}
X\times G&\longrightarrow X\\
(x,g)&\longmapsto xg
\end{align*}
such that $x1_G=x$ and $x(gg')=(xg){g'}$ holds for all $x\in X$ and $g,g'\in G$. The fundamental property is that 
\[
x \sim_G x' :\iff \exists g\in G: x'=xg
\]
defines an equivalence relation on $X$. The induced equivalence classes are the \emph{orbits} of $G$ on $X$. The orbit of $x\in X$ is abbreviated by 
\[
xG:=\{ xg\mid g\in G \}
\]
and we denote the set of all orbits of $G$ on $X$ by
\[
X/G:=\{ xG\mid x\in X\}.
\]
A transversal of the orbits $X/G$, denoted by ${\cal T}(X/G)$, is a minimal subset of $X$ such that
$X/G:=\{ xG\mid x\in {\cal T}(X/G)\}$, i.\,e. it is a set of representatives of the orbits.

Another useful tool in order to determine the corresponding orbit of an element $x\in X$ is the \emph{canonizing mapping}:
\begin{align*}
\gamma_{{\cal T}(X/G)}:X &\longrightarrow G \\
x &\longmapsto g\text{ with }xg\in {\cal T}(X/G).
\end{align*}
Two elements $x,x'\in X$ are in the same orbit, i.e. $xG=x'G$, if and only if
\[
x\gamma_{{\cal T}(X/G)}(x)=x'\gamma_{{\cal T}(X/G)}(x').
\]

The \emph{stabilizer} of an element $x\in X$ is the set of group elements that fix $x$:
\[
\stab{G}{x}:=\{g\in G\mid xg=x\}.
\]
Stabilizers are subgroups of $G$ having the property that the stabilizers of different elements of the same orbit are conjugated subgroups:
\[
\stab{G}{xg}=g^{-1}\stab{G}{x}g.
\]
The \emph{Fundamental Lemma of group actions} \cite{DM96,Ker99} says, that an orbit can be bijectively mapped onto the right cosets of the stabilizer:
\begin{align*}
\stab{G}{x}\backslash G &\rightarrowtail\hspace{-2ex}\rightarrow xG \\
\stab{G}{x}g &\mapsto xg.
\end{align*}
In particular, if ${\mathcal T}(\stab{G}{x}\backslash G)$ denotes a transversal between $\stab{G}{x}$ and $G$, the mapping ${\mathcal T}(\stab{G}{x}\backslash G)\to xG, g\mapsto xg$ is also one-to-one. As an immediate consequence  we obtain for the orbit size the equation:

\begin{prop}\label{prop5}
\[
|xG|=\genfrac{}{}{}{}{|G|}{|\stab{G}{x}|}
\]
\end{prop}

\subsection{Group actions and codes over metric sets}\label{sec:orbitmetric}

Now we switch to coding theory and assume $X$ to be a set admitting a metric function $d:X\times X\to {\mathbb R}$.  A \emph{code} is a subset ${\cal C}\subseteq X$. Its minimum distance is defined as
\[
d({\cal C}) :=\min\{d(x,x')\mid x,x'\in {\mathcal C}, x\ne x'\} .
\]

A \emph{minimum distance decoder} of $\mathcal{C}\subseteq X$ is a maximum likelihood decoder 
\[
\dec_{\mathcal{C}}:X\longrightarrow \mathcal{C},
\]
which maps an element $r\in X$ onto an element $c\in {\cal C}$ such that $d(c,r)\le d(c',r)$ for all $c'\in {\cal C}$.

The following property is well-known:

\begin{lem}
Let $\mathcal {C}\subseteq X$, $c\in \mathcal{C}$, and let $r\in X$. 
If $d(c,r)\le \lfloor\genfrac{}{}{}{}{d({\cal C})-1}{2}\rfloor$, the minimum distance decoder $\dec_{\mathcal{C}}$ applied to $r$ finds the uniquely determined element $c$.
\end{lem}

From now on assume that the group elements act as an isometry on $X$, i.\,e. 
\[
d(xg,x'g)=d(x,x')
\]
for all $x,x'\in X$ and $g\in G$.
For two subsets $A,B\subseteq X$ we define the \emph{intersubset distance} as
\[
\hat d(A,B)=\min\{d(x,y)\mid x\in A, y\in B\}
\]
which in fact defines a metric on the set of subsets of $X$. Using this intersubset distance we can define a metric on the set of orbits of $G$ on $X$:
\begin{align*}
\hat d(xG,yG)&=\min\{d(xg,yh)\mid g,h\in G\}\\
&=\min\{d(x,yhg^{-1})\mid g,h\in G\}\\
&=\min\{d(x,yg')\mid g'\in G\}   .
\end{align*}

\begin{lem}\label{lem:hatdelta:metric}
The mapping $\hat d$ defines a metric function on the set of orbits $X/G$.
\end{lem}

\begin{proof}
The properties symmetry and positive definiteness of $\hat d$ follow directly from the properties of the metric $d$. The triangle inequality also holds: Let $x,y,z\in X$ and let $g,h\in G$ such that $\hat d(xG,yG)=d(x,yg)$ and $\hat d(yG,zG)=d(y,zh)=d(yg,zhg)$. Then we get
\begin{align*}
\hat d(xG,yG)+\hat d(yG,zG)
&=d(x,yg)+d(yg,zhg)\\
&\ge d(x,z(hg))\\
&\ge \hat d(xG,zG)
\end{align*}
which completes the proof.
\end{proof}

\begin{defi}
Codes that are orbits of a group $G$ on a metric set $X$, are called \emph{orbit codes}. If the group $G$ is cyclic we call them \emph{cyclic orbit codes}.
\end{defi}

We can evaluate the minimum distance of orbit codes more conveniently than considering all pairs of elements.

\begin{thm}\label{lem:mindist}
Let ${\cal C}=xG$ for an element $x\in X$. Then
\[
d({\cal C})=\min\{d(x,xg)\mid g\in{\mathcal T}(\stab{G}{x}\backslash G),g\not  \in \stab{G}{x}\}.
\]
\end{thm}

\begin{proof}
To evaluate the minimum distance we have to consider all $d(y,z)$ for all $y,z\in {\cal C}$ with $y\ne z$, i.\,e. $d(xg,xh)=d(x,xhg^{-1})$ for all $g,h\in G$ such that $xg\ne xh$, or equivalently $d(x,xg)$ for all $g\in G$ with $xg \ne x$. The Fundamental Lemma yields that it is sufficient to consider all $g$ from a transversal between the stabilizer $\stab{G}{x}$ and $G$ in order to run through the orbit $xG$.
\end{proof}

The following lemma describes a minimum distance decoder for orbit codes in terms of group actions using the canonizing mapping. The crucial fact is, that we have a transversal ${\cal T}(X/G)=\{x_0,x_1,\ldots,x_{\ell-1}\}$ of the orbits of $G$ on the set $X$, such that the distance of two transversal elements $x_0$ and $x_i$ is the minimal distance between elements of the orbit code ${\mathcal C}=x_0G$ and the orbit $x_iG$, i.\,e. the distance between $x_0$ and $x_i$ is the intersubset distance between the two sets $x_0G$ and $x_iG$.

\begin{thm}\label{lem:can:decoding}
Let ${\cal T}(X/G)=\{x_0,x_1,\ldots,x_{\ell-1}\}$ be a transversal of the orbits, such that $\hat d(x_0G,x_iG)=d(x_0,x_i)$ for all  $1\le i< \ell$ and let ${\cal C}=x_0G$ be the corresponding orbit code. Then the mapping
\begin{align*}
\dec_{\cal C}:X &\longrightarrow \Cvs\\
r&\longmapsto x_0\gamma_{{\cal T}(X/ G)}(r)^{-1}
\end{align*}
yields a minimum distance decoder, i.\,e. $x_0\gamma_{{\cal T}(X/ G)}(r)^{-1}\in \mathcal{C}$ is the closest codeword to $r\in X$.
\end{thm}

\begin{proof}
Assume that $r$ is in the orbit of $x_i$. The aim is now to find a codeword $c\in \mathcal{C}=x_0G$, such that $d(c,r)$ is minimal. Let  $g:=\gamma_{{\cal T}(X/ G)}(r)$ be the group element that maps $r$ onto the orbit representative $x_i$, i.\,e. $x_i=rg$. Then we obtain 
\begin{align*}
d(x_0,x_i)&=d(x_0,rg)=d(x_0g^{-1},r)\\&=d(x_0\gamma_{{\cal T}(X/G)}(r)^{-1},r) .
\end{align*}
Since we know that $d(x_0,x_i)$ is minimal between elements of $\mathcal{C}=x_0G$ and $x_iG$ we get that $d(x_0\gamma_{{\cal T}(X/G)}(r)^{-1},r)$ is also minimal. Hence $c:=x_0\gamma_{{\cal T}(X/G)}(r)^{-1}$ is the closest codeword to $r$.

This fact can easily be verified using Figure~\ref{fig:vis:dec4}.
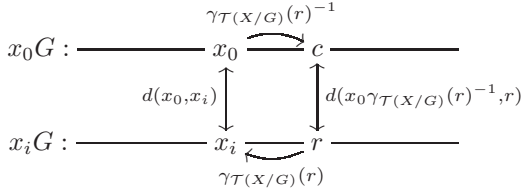
\begin{figure}[!htbp]
\[
\SelectTips{eu}{12} 
\xymatrix @-3pt
{
x_0G:\ar@{-}[rr]
&
&{x_0}\ar@{<->}[d]_{d(x_0,x_i)}\ar@/^/[r]^{\gamma_{{\cal T}(X/G)}(r)^{-1}}\ar@{-}[r]
& c\ar@{<->}[d]^{d(x_0\gamma_{{\cal T}(X/G)}(r)^{-1},r)}\ar@{-}[rr]
&
&\\
x_iG:\ar@{-}[rr]
&
&{x_i} \ar@{-}[r]
& r\ar@/^/[l]^{\gamma_{{\cal T}(X/G)}(r)}\ar@{-}[rr]
&
&\\
}
\]
\caption{Visualization of the minimum distance decoder}\label{fig:vis:dec4}
\end{figure}
\end{proof}


\subsection{Isometric orbit codes}\label{sec:iso}

Let $H$ be a distance-preserving group acting \emph{transitively} on the set $X$, i.\,e. $xH = X$ for all $x\in X$.

In this section we want to classify orbit codes that arise by subgroups $G$ of $H$, denoted by $G\leq H$. For this purpose we introduce the notion of isometry of orbit codes:

\begin{defi}
Two orbit codes $\mathcal{C}$ and $\mathcal{C}'$ defined by subgroups $G,G'\le H$, i.\,e. $\mathcal{C}=xG$ and $\mathcal{C}'=x'G'$ for elements $x,x'\in X$, are called \emph{$H$-isometric}, if there is a group element $h\in H$ mapping $\mathcal{C}$ onto $\mathcal{C}'$:
\[
\mathcal{C}'=\mathcal{C}h:=\{xh\mid x\in\mathcal{C}\}
\]
\end{defi}

\begin{thm}\label{thm:isometricorbitcodes} The following properties hold:
\begin{enumerate}
\item
Let $G\leq H$ and $G'=h^{-1}Gh$ for some $h\in H$, and $x\in X$. Then for $x'=xh$ the orbit codes $\mathcal{C}=xG$ and $\mathcal{C}'=x'G'$ are $H$-isometric.
\item 
Let $\mathcal{C}=xG$ and $\mathcal{C}'=\mathcal{C}h$ for some $h\in H$, i.e. $\Cvs$ and $\Cvs'$ are $H$-isometric. Then $\mathcal{C'}$ is an orbit code $\mathcal{C}'=x'G'$ with $x'=xh$ and $G'=h^{-1}Gh$.
\end{enumerate}
\end{thm}

\begin{proof}
Both statements follow from 
\begin{align*}
\mathcal{C}'&=\{x'g'\mid g'\in G'\}\\
&=\{xh(h^{-1}gh)\mid g\in GÊ\}\\
&=\{(xg)h\mid g\in G\}\\
&= \mathcal{C}h    .
\end{align*} 
\end{proof}

\begin{co}
Two orbit codes are $H$-isometric, if and only if they arise from conjugate subgroups of $H$.
\end{co}

Hence, it is sufficient to consider just one representative of each conjugacy class of \emph{subgroups} of $H$, in order to classify by $H$-isometry. In the case of cyclic orbit codes we can even restrict to conjugacy classes of \emph{group elements} of $H$.


\begin{lem} Two cyclic subgroups $G,G'$ of $H$ are conjugate if and only if there exist some generators $g,g'\in H$ of these groups, i.e. $G=\langle g \rangle$ and $G'=\langle g' \rangle$, which are conjugate, i.\,e. $g'=h^{-1}gh$ for some $h\in H$.
\end{lem}

\begin{proof}
\begin{description}
\item[\fbox {$\Leftarrow$}]
Let $g'=h^{-1}gh$. Then
\begin{align*}
\langle g'\rangle &=\langle (h ^{-1} g h) \rangle\\
&=\{(h ^{-1} g h)^i\mid 0\le i <\ell\}\\
&=\{h ^{-1} g^i h\mid 0\le i <\ell\}\\
&=h^{-1}\langle g \rangle h
\end{align*}
i.e. $G'=h^{-1}Gh$.
\item[\fbox {$\Rightarrow$}]
Let $G'=h^{-1}Gh$ and $g$ a generator of $G$. Then
\begin{align*}
G'&=h^{-1}Gh\\
&=\{h^{-1}g^ih\mid 0\le i < \ell\}\\
&=\{(h^{-1}gh)^i\mid 0\le i < \ell\}\\
&=\langle h^{-1}gh \rangle
\end{align*}
i.e. $h^{-1}gh$ is a generator of $G'$.
\end{description}
\end{proof}

\begin{co} 
Two cyclic orbit codes are $H$-isometric if and only if they arise by two cyclic subgroups $G_{1}, G_{2}\le H$ that are generated by conjugate elements of $H$. I. e. there exist $g_{1}, g_{2} \in H$ such that
\[G_{1}=\langle g_{1}\rangle , \quad G_{2}=\langle g_{2} \rangle \]
and $g_{2}=h^{-1}g_{1}h$ for some $h\in H$.
\end{co}

In order to furthermore study their properties, we introduce the notion of distance distribution of orbit codes which is adapted from the definition of  weight enumerators of classical block codes.

\begin{defi}
Let $\mathcal{C}=xG$ be an orbit code. The \emph{distance distribution} of $\mathcal{C}$ is the sequence $D_0,D_1,D_2,\ldots$ of natural numbers with
\[
D_i := |\{ y\in{\mathcal C}\mid d(x,y)=i\}|   .
\]
\end{defi}

%

\begin{thm} 
$H$-isometric orbit codes have the same distance distribution.
\end{thm}

\begin{proof}
Let $\Cvs=xG$ and $\Cvs'=\Cvs h$ (for some $h\in H$) be two $H$-isometric orbit codes. Hence $h$ is a bijection between the two codes. Moreover, 
\[d(xh, yh) = d(x, y)
\]
for and $x,y \in \Cvs$, i.e. the number of elements in $\Cvs$ with a fixed distance from $x$ is the same as the number of elements in $\Cvs'$ with that distance from $xh$.
\end{proof}

\subsection{Characterization of linear codes as orbit codes}

Let $X={\mathbb F}_q^n$ and let $C\le {\mathbb F}_q^n$ be a linear code of dimension $k$. For any vector $v\in {\mathbb F}_q^n$ the mapping 
\begin{align*}
\tau_v:{\mathbb F}_q^n&\longrightarrow {\mathbb F}_q^n\\
x&\longmapsto x\tau_v:=x+v
\end{align*}
defines a bijection. Since $C$ is an additive group the set
\[
G=\{\tau_c\mid c\in C\}
\]
also forms a group with respect to the composition. The map
\begin{align*}
{\mathbb F}_q^n\times G &\longrightarrow {\mathbb F}_q^n\\ 
(x,\tau_c) &\longmapsto x\tau_c=x+c
\end{align*}
defines a group action, preserving the Hamming distance:
\[
d(x,y)=d(x+c,y+c)=d(x\tau_c,y\tau_c).
\]

Since the stabilizers $G_x$ of all elements $x\in{\mathbb F}_q^n$ are trivial, and the group has size $|G|=q^{k}$ the number of orbits is
\[
\ell=|{\mathbb F}_q^n/G|=q^{n-k}.
\]

Let ${\cal T}({\mathbb F}_q^n/G)=\{x_0,x_1,\ldots,x_{\ell-1}\}$ be a transversal with $x_0=0$ satisfying the desired property
\[
\hat d(x_0G,x_iG)=\hat d(0G,x_iG)=d(0,x_i)=\wt(x_i)
\]
for all $1\le i< \ell$, i.\,e. the orbit representative $x_i$ has minimal weight within its orbit. Such an element is called \emph{coset leader}. 

If we take $x_0=0$ the corresponding orbit code ${\cal C}=x_0G=0G$ satisfies
\[
{\cal C}=0G=\{0{\tau_c}\mid c\in C\}=\{c\mid c\in C\}=C,
\]
i.\,e. it is the linear code itself. Since the stabilizers $\stab{G}{x}$ for all $x\in \F_q^n$ are trivial, a transversal between $G$ and $\stab{G}{0}$ is the complete set $G$, i.\,e. ${\mathcal T}(\stab{G}{0}\backslash G)=G$. Then Lemma \ref{lem:mindist} yields the formula:
\begin{align*}
d({\cal C})&=\min\{d(0,0\tau_c)\mid \tau_c\in {\mathcal T}(\stab{G}{0}\backslash G),\tau_c\ne id\}\\
&=\min\{d(0,c)\mid c\in C, c\ne 0\}.
\end{align*}
If $r\in x_iG$ the element $\tau_{x_i-r}\in G$ is the canonical element of $x_{i} G$, i.\,e.
\[
\gamma_{{\mathcal T}({\mathbb F}_q^n/G)}(r)=\tau_{x_i-r},
\]
since 
\[
y\tau_{x_i-r}=r+(x_i-r)=x_i.
\] 
It is obvious that $\tau_{x_i-r}^{-1}=\tau_{r-x_i}$. Hence, according to Lemma \ref{lem:can:decoding} the mapping
\begin{align*}
y\mapsto \, x_0\gamma_{{\mathcal T}({\mathbb F}_q^n/G)}(r)^{-1} &= x_0{\tau_{x_i-r}^{-1}}\\&=0{\tau_{r-x_i}}\\&=0+(r-x_i)\\&=r-x_i
\end{align*}
defines the minimum distance decoder. In fact, this corresponds exactly to the decoding procedure, which is known as \emph{coset leader decoding}. The challenge here is to find the corresponding orbit $x_iG$ and hence the corresponding orbit representative $x_i$, which has minimal weight in its orbit. The determination of the orbit can be done by the evaluation of the syndrome, which is invariant on each orbit: If $\Delta$ denotes a check matrix of the linear code $C$ the equivalence holds:
\[
x\Delta^t=y\Delta^t\iff xG=yG.
\]

We now examine the $H$-isometry of these codes, where
\[H=\{\tau_{c} \mid c \in \F_{q}^{n}\}\] 
acts on $\F_{q}^{n}$ the same way that $G$ does. 
According to Theorem \ref{thm:isometricorbitcodes} a code $\mathcal{C}'=x'G'$ with 
\[x'=0\tau_{h}=h , \quad G'=\tau_{h}^{-1} G \tau_{h}\]
 is $H$-isometric to the orbit code $\mathcal{C}=0G$. Since the group $H$ is commutative we get $G'=G$ and hence 
\[
\mathcal{C'}=hG=\{h+c\in C\} = h+C
\] 
if $G=\{\tau_c\mid c\in C\}$ is defined by a linear code $C\le \F_q^n$. Finally, the codes which are $H$-isomorphic to a linear code $C\le \F_q^n$ are the affine spaces $h+C$ for any $h\in \F_q^n$.

\subsection{Random network orbit codes}

Now we consider random network codes arising from group actions which are simply called \emph{orbit codes} and which were introduced in \cite{tr10p}.

The general linear group of degree $n$, denoted by $\GL_n$, is the set of all invertible $n\times n$-matrices with entries in $\F_{q}$. If we have to specify the underlying field we will write $\GL_n(\F_q)$:
\[
\GL_n(\F_q):=\{A\in \F_q^{n\times n}\mid \rank(A)=n\}.
\] 

In order to represent a $k$-subspace $\mathcal U\in \Gr$ we use a \emph{generator matrix} which is a $k\times n$-matrix $U$ whose rows form a basis of $\mathcal U$, i.\,e.
\[
\mathcal{U} = \rs (U) :=\textrm{row space}(U)= \{ x U\mid x \in \F_{q}^k\}.
\]

Multiplication with $\GL_n$-elements actually defines a
group action from the right on the projective space $\PS$ by
\begin{align*}
\PS\times \GL_n &\longrightarrow \PS\\
(\mathcal{U},A) &\longmapsto \mathcal{U} A:=\{vA\mid v\in {\mathcal U}\}.
\end{align*}

Since any invertible matrix maps subspaces to subspaces of the same dimension and since two $k$-subspaces can be mapped onto each other by an invertible matrix, the orbit of the general linear group $\GL_n$ on any $k$-subspace $\mathcal{U}$ is the whole set of all $k$-subspaces. Thus, $\GL_n$ acts transitively on $\Gr$. In particular, if $\mathcal{U}_k$ denotes the standard $k$-subspace generated by the first $k$ unit vectors, we get
\[
\mathcal{U}_k\GL_n= \Gr.
\]
Hence, the set of all orbits is
\[
\PS/\GL_n = \{ \mathcal{G}_q(0,n), \mathcal{G}_q(1,n),\ldots,\mathcal{G}_q(n,n)\}.
\]

It is easy to verify that the stabilizer $\stab{\GL_n}{\mathcal{U}_k}$ of the standard $k$-subspace $\mathcal{U}_k$ consists of all matrices of the form
\[
\left(\begin{array}{c|c}
X & 0\\
\hline
Y & Z
\end{array}
\right)
\]
where $X\in \GL_k$, $Y\in\F_q^{n-k\times k}$ and $Z\in \GL_{n-k}$.
Since every $k$-subspace ${\mathcal U}$ lies in the orbit of $\mathcal{U}_k$, say $\mathcal{U}=\mathcal{U}_kA$ for some $A\in \GL_n$, we get
\[
\stab{\GL_n}{\mathcal{U}}=A^{-1}\stab{\GL_n}{\mathcal{U}_k}A.
\]

If we take a subgroup ${G}$ of the general linear group $\GL_n$ and consider the orbits of ${G}$ on the set of all subspaces $\PS$ we obtain that the $GL_n$-orbits split into the ${G}$-orbits, such that $G$ induces an action on the Grassmannian $\Gr$:
\begin{align*}
\Gr\times {G} &\longrightarrow \Gr\\
(\mathcal{U},A) &\longmapsto \mathcal{U}A.
\end{align*}
The corresponding stabilizers of this action are 
\[
\stab{{G}}{\mathcal{U}}={G}\cap\stab{\GL_n}{\mathcal{U}}.
\]

Now the definition of an orbit code in the setting of the projective space is straight-forward. 

\begin{defi}
The orbits of a subgroup of the general linear group $\GL_n$ on the Grassmannian $\Gr$ are called \emph{(subspace) orbit codes}.
\end{defi}

From Proposition \ref{prop5} and Theorem \ref{lem:mindist} we can deduce:

\begin{thm}\label{thm:basic}
Let $\mathcal{U}\in \Gr$, $G\le \GL_n$ and ${\mathcal C} ={\mathcal U}{G}$ be an orbit code. 
\begin{enumerate}

\item 
The code ${\mathcal C}$ has size
\[
|{\mathcal C}|=\genfrac{}{}{}{}{|{G}|}{|\stab{{G}}{\mathcal{U}}|}=\genfrac{}{}{}{}{|{G}|}{|{G}\cap\stab{\GL_n}{\mathcal{U}}|}.
\]

\item The minimum distance $d(\mathcal{C})$ of the code satisfies
\begin{align*}
d(\mathcal{C})=&\min\{d(\mathcal{U},\mathcal{U}A)\mid \\
 &A\in \mathcal{T}(\stab{{G}}{\mathcal{U}}\backslash{G}), A\not \in \stab{G}{\Uvs}\}.
\end{align*}

\item 
The stabilizers in $GL_n$ of different codewords $\mathcal{V},\mathcal{W}\in{\mathcal C}$ are conjugated subgroups, i.\,e. there exists $A\in {G}$ with
\[
\stab{\GL_n}{\mathcal{V}} = A^{-1}\stab{\GL_n}{\mathcal{W}}A.
\]
and
\[
\stab{{G}}{\mathcal{V}} = A^{-1}\stab{{G}}{\mathcal{W}}A.
\]
\item
In particular, 
$
|\stab{\GL_n}{\mathcal{V}}|=|\stab{\GL_n}{\mathcal{W}}|$, 
respectively
$
|\stab{{G}}{\mathcal{V}}|=|\stab{{G}}{\mathcal{W}}|.
$
 \end{enumerate}
\end{thm}

As mentioned in Section \ref{sec:background}, the dual code $\mathcal{C}^\perp=\{\mathcal{U}^\perp\mid \mathcal{U}\in \mathcal{C}\}$ also defines a code with the same minimum distance as the original code. Moreover, the dual of an orbit code forms an orbit code as well \cite{tr10p}:

\begin{thm} 
Let $\mathcal{U}\in \Gr$, ${G}$ a subgroup of $\GL_n$ and ${\mathcal C} ={\mathcal U}{G}$ be an orbit code. Then the dual satisfies $\mathcal{C}^\perp=(\mathcal{U}^\perp)G^t$ where $G^t=\{A^t\mid A\in G\}$.
\end{thm}

\begin{proof}
The proof immediately follows from the identity $(\mathcal{U}A)^{\perp}=(\mathcal{U}^{\perp})(A^{-1})^t$.
\end{proof}

Denote by $D_{0}(\Cvs), D_{2}(\Cvs), \dots ,D_{2k}(\Cvs)$ the distance distribution of the orbit code $\Cvs\subseteq \Gr$ and define the distance enumerator polynomial of an orbit code $\Cvs\subseteq \Gr$ by
\[\mathcal{D}_{\Cvs}(x,y):= \sum_{i=0}^{k} D_{2i}(\Cvs) x^{i}y^{k-i} .\]
 Then one can easily derive the analogue of the MacWilliams identity:

\begin{thm}
\[\mathcal{D}_{\Cvs}(x,y)= \mathcal{D}_{\Cvs^{\perp}}(x,y) \]
\end{thm}

\begin{proof}
The statement follows from the fact that $d(\Uvs^{\perp}, \Vvs^{\perp})=d(\Uvs, \Vvs)$ for any $\Uvs, \Vvs \in \Gr$, 
which implies that $D_{2i}(\Cvs)=D_{2i}(\Cvs^{\perp})$ for all $i=0,\dots,k$.
\end{proof}

\section{Classification of Cyclic Orbit Codes}
\label{sec:classification}

As described in section \ref{sec:iso} isometric network orbits codes
arise from the different conjugacy classes of subgroups of the general
linear group $H=\GL_n$. Moreover, the isometric cyclic orbit codes
correspond to the different conjugacy classes of elements of
$\GL_n$. A canonical transversal of all these conjugacy classes can be
obtained from the different \emph{rational canonical forms} of
invertible matrices. In this section we want to specify these matrix
representatives and give some properties of the groups and orbit
codes generated by these matrices in canonical forms.

\subsection{Rational canonical forms and conjugacy classes}
%

For simplicity, in this section we abbreviate the notation
of polynomials to $p\in\F_q[x]$. 


\begin{defi}\label{d:com_mat}
Let $p=p_0+p_1x+\ldots+p_{n-1}x^{n-1}+x^n\in F_{q}[x]$. The
\emph{companion matrix} of $p$ is
\[
M_p:=\mat{
0&1&0 &\cdots &0\\
0 & 0 & 1 & &0\\
\vdots & & & \ddots &\vdots\\
0&0&0& & 1 \\ 
-p_0& -p_1 & -p_2 & \cdots & -p_{n-1}
}\in \F_q^{n\times n}
.\]
\end{defi}

\begin{defi}
Let $p(x)\in \F_q[x]$ with $p(0)\neq 0$. Then the least integer $e\in \N$ such that $p(x)$ divides $x^e-1$ is called the \emph{order} of $p(x)$.
\end{defi}

The interested reader can find more information related to the previous
definitions in \cite{li86}. 

We characterize conjugacy classes of elements of $\GL_n$ by the rational canonical forms. The following result can be found in \cite[Chapter~6.7]{he75}.

\begin{thm}
  Let $A\in \GL_n$. Then there exist a transformation matrix $L\in
  \GL_n$, distinct monic irreducible polynomials $p_1,\ldots,p_m\in
  \F_q[x]$  and integer
  partitions $e_1,\ldots, e_m$ of natural numbers $n_1,\ldots,n_m$,
  $n=\sum_i n_i$,
\[
e_i=(e_{i,1},\ldots,e_{i,r_i})\text{ with } e_{i,1}\ge\ldots\ge
e_{i,r_i}\text{ and } n_i=\sum_{j}e_{i,j}
\]
such that
\begin{align*}
  &L^{-1}AL=\\
  &\quad\quad\diag(M_{p_1^{e_{1,1}}},\ldots, M_{p_1^{e_{1,r_1}}},
  \ldots,M_{p_m^{e_{m,1}}}, \ldots, M_{p_m^{e_{m,r_m}}})
\end{align*}
forms a block diagonal matrix with companion matrices $M_{p_i^{e_{i,j}}}$ of
the polynomials $p_i^{e_{i,j}}$. Moreover, this matrix is unique (up
to ordering of the polynomials) for any choice of $A\in \GL_n$.  
Hence, we call this block diagonal matrix the \emph{rational canonical
  form} of $A$ and we write
\[
\RCF(A):=L^{-1}AL.
\]
The polynomials $p_i^{e_{i,j}}$ are called \emph{elementary divisors}
of $A$.
\end{thm}

\begin{co}\label{cor:rcf:conjugated}
Two matrices $A,B\in \GL_n$ are conjugate if and only if $\RCF(A)=\RCF(B)$.
\end{co}

The conjugacy classes of elements in $\GL_n$ are determined by the
finite lists of $m$ distinct monic irreducible polynomials
$p_1,\ldots,p_m$ and the $m$ partitions $e_1,\ldots,e_m$ of the
numbers $n_1,\ldots,n_m$ such that $n=\sum_i n_i$.

If we consider conjugation of subgroups, which is what matters for the
consideration of isometry of orbit codes, we show in the following
that the number $m$ and the partitions $e_1,\ldots,e_m$ and the order
of the polynomials $p_1,\dots,p_m$ are invariant
on the generators of subgroups of the same conjugacy class of
subgroups. Hence, the choice of irreducible polynomials
$p_1,\ldots,p_m$ is based only on their order.

More formally, assume that the list of partitions $e_1,\ldots,e_m$ is
in a unique ordering (e.\,g. lexicographically increasing), and let
$o_1,\ldots, o_m\in \N$ be the orders of the polynomials
$p_1,\ldots,p_m$, then the couple given by the ordered vector of all partitions 
\[
e(A):=(e_1,\ldots,e_m) 
\]
and by the vector \[o(A):=(o_1,\ldots,o_m)\] will be called the
\emph{matrix type} of the rational canonical form of a matrix $A$.

\begin{co} \label{cor:conj:sametype} 
  Let $A,B\in \GL_n$ be conjugate
  matrices. Then the rational canonical forms of both matrices have
  the same type, i.e.
  \[e(A)=e(B) \mbox{ and } o(A)=o(B).\]
\end{co}

\begin{lem}\label{lem:group:sametype}
The matrix type is constant on the set of all generators of the same cyclic group:
\[
\langle A\rangle = \langle B\rangle \,\Longrightarrow\, e(A)=e(B)
\mbox{ and } o(A)=o(B).
\]
\end{lem}

\begin{proof}
  First we show the case where $A$ has a unique elementary divisor. At
  the end of the proof we will give the main remark that implies the
  generalized statement.

  From the unique divisor of $A$, it follows that the characteristic
  and the minimal polynomial of $A$ are both $p^{e}\in
  \F_q[x]$ where $p\in \F_q[x]$ is irreducible. Define $k:=n/e$ and
  let $\F_{q^k}:=\F_q[x]/(p)$ be the splitting field of the polynomial
  $p$ and $\mu\in \F_{q^k}$ a primitive element of it. There exists a
  $j\in \N$ such that $p=\prod_{u=0}^{k-1}(x-\mu^{jq^u})$. Since
  $p^{e}$ is the characteristic and the minimal polynomial of $A$, we
  obtain that the Jordan normal form of $A$ over $\F_{q^k}$ is
  \[
  J_A=\diag\left(J_{A,\mu^j}^{e},\dots,J_{A,\mu^{jq^{k-1}}}^{e}\right)
  \]
  where $J_{A,\mu^{jq^u}}^{e}\in \GL_{e}(\F_{q^k})$ is a unique Jordan
  block with diagonal entries $\mu^{jq^u}$ for $0\le u<k$.

  By the Jordan normal form of $A$ it follows that for every $i\in \N$
  the characteristic polynomial of $A^i$ is
  \[p_{A^i}=\left(\prod_{u=0}^{k-1}x-\mu^{ijq^u}\right)^{e}.\] Since $B$
  is a generator of $\langle A\rangle$, there exists an $i$ with
  $\gcd(i,|\langle A\rangle|)=1$ such that $B=A^i$. It follows that
  $p_{A^i}\in \F_q[x]$ is a monic irreducible polynomial whose order is
  the same as the one of $p$.

  In order to conclude that $p_{A^i}^{e}$ is the unique elementary
  divisor of $A^i$, we consider its rational canonical form. Without
  loss of generality, assume the rational canonical form of $A^i$ to
  be
  \[\RCF(A^i)=\diag\left(M_{p_{A^i}^{e_{1}}},M_{p_{A^i}^{e_{2}}}\right)\] where
  $e=e_{1}+e_{2}$. For any $j\in \N$ we obtain that the matrix
  $\RCF((\RCF(A^i))^j)$ is a block diagonal matrix with at least two
  blocks. Let $j\in \N$ such that $ij\equiv 1 \pmod{|\langle
    A\rangle|}$ and $L\in \GL_n$ be a matrix such that
  $\RCF(A^i)=L^{-1}A^iL$, then
\[
(\RCF(A^i))^j=(L^{-1}A^iL)^j=L^{-1}AL
\]
implying that
\[
\RCF(A)=\RCF((\RCF(A^i))^j) .
\] 
This leads to a contradiction since $\RCF(A)=M_{p^{e}}$ has only one
block. We conclude that $p_{A^i}^{e}$ is the elementary divisor of $A^i$.

If $m>1$ the only difference is the choice of the splitting
field. 
If $p_1,\dots,p_m$ are the distinct monic irreducible polynomials
related to $\RCF(A)$, then the splitting field on which the proof is
based is $\F_q[x]/(\prod_{t=1}^m p_i)$.
\end{proof}

Due to this result, the following definition is well-defined. If
$G=\langle A\rangle$ denotes a cyclic subgroup of $\GL_n$, the type of
the cyclic group $G$ is defined by
\[
e(G):=e(A) \mbox{ and } o(G):=o(A).
\]

\begin{thm}\label{t:rep_grp}
Let $G,G'$ be two cyclic subgroups of $\GL_n$. Then they are conjugate if and only if they have the same group type:
\[
G,G^\prime\text{ conjugate} \iff e(G)=e(G^\prime) \mbox{ and } o(G)=o(G^\prime).
\]
\end{thm}

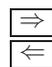
\begin{proof}
\begin{description}
\item[\fbox{$\Rightarrow$}] Follows from Corollary
  \ref{cor:conj:sametype} and Lemma \ref{lem:group:sametype}.

\item[\fbox{$\Leftarrow$}] Let $G=\langle A\rangle$ and
  $G^\prime=\langle A^\prime\rangle$, we obtain that
  \[e(A)=e(A^\prime)\mbox{ and }o(A)=o(A^\prime).\] Let
  $p_{A,1},\ldots p_{A,m},p_{A^\prime,1},\ldots
  p_{A^\prime,m},\in\F_q[x]$ be the bases of the elementary divisors of $A$ and $A'$, respectively.  Since $o(A)=o(A^\prime)$, let $\F$
  be the unique splitting field of $(\prod_{t=1}^rp_{A,t})(\prod_{t=1}^rp_{A^\prime,t})$, and $\mu\in \F$ a primitive element
  of it.

  Let $d_j:= \deg p_{A,j}$ for $j=1,\ldots,
  m$. Then, there exist  $i_1,\dots,i_m\in \N$  such that
  \[p_{A,j}=\prod_{u=0}^{d_j-1}(x-\mu^{i_j q^u})\] for
  $j=1,\dots,m$. The same holds for the matrix $A^\prime$, i.\,e. there exist
  $i^\prime_{1},\dots,i^\prime_{m}\in \N$ such that
  \[p_{A^\prime,j}=\prod_{u=0}^{d_j-1}(x-\mu^{i^\prime_jq^u})\] for
  $j=1,\dots,m$. By the condition on the orders, there exists a unique
  $i\in \N$ such that \[i_j\equiv i\cdot i^\prime_j
  \pmod{o_j}\] for $j=1,\dots,m$. We obtain that
  $p_{(A^\prime)^i,j}=p_{A,j}$ which combined with the condition
  $e(A)=e(A^\prime)$ implies that $\RCF((A^\prime)^i)=\RCF(A)$,
  i.\,e. $A$ and $(A^\prime)^i$ are conjugate.
\end{description}
\end{proof}

In order to obtain a transversal of the conjugacy classes of all
cyclic subgroups of $\GL_n$ we have to consider their group types. These correspond to all different ordered lists of
partitions $(e_1,\ldots,e_m)$, $e_i=(e_{i,1},\ldots,e_{i,r_i})$, such
that $\sum_{i,j}e_{i,j}=n$ together with possible orders $(o_1,\ldots,o_m)$.

Hence, we can uniquely represent the conjugacy class of a cyclic subgroup of $\GL_n$ by a rational canonical form where the irreducible polynomials of the elementary divisors have the given degrees and orders.




\subsection{Cyclic orbit codes from block matrices}\label{sec:conj_and_codes}

In order to simplify the following statements, let
$M:=\diag(M_1,\dots,M_t)\in \GL_n$ be a block diagonal matrix where the
$M_i$ are the companion matrices of the polynomials $p_i^{e_i}\in
\F_q[x]$, the $p_i$ are monic irreducible polynomials, and
$n_i=\deg(p_i^{e_i})$.

\begin{lem}\label{l:cardd}
  Let $U=\rs\mat{U_1,\dots,U_t}\in \Gr$, $r_i:=\min\{r\in \N\mid
  \rs(U_i)M_i^r=\rs(U_i)\}$ and $s_i:=\min\{s\in \N\mid
  U_iM_i^s=U_i\}$ for $i=1,\dots,t$. Then the smallest $g\in \N$ with
   \[\Stab_{
    \langle M\rangle}(\Uvs)=\langle M^g\rangle\] 
    fulfills 
    \[\mathrm{lcm}(r_i\mid
  i\in \{1,\dots, t\})\mid g\mid \mathrm{lcm}(s_i\mid i\in
  \{1,\dots,t\}).\]
\end{lem}

\begin{proof}
  Let \[g:=\min\left\{l\in \{0,\dots, |\langle M\rangle|-1\}\mid
  \Stab_{\langle M\rangle}(\Uvs)=\langle M^l \rangle\right\}.\] 
  It holds that
\begin{enumerate}
\item $\mathrm{lcm}(r_i\mid i\in \{1,\dots, t\})\mid g$: Since $M^g\in
  \stab{\langle M\rangle}{\Uvs}$, then there exists an $N\in \GL_k$
  such that
  \begin{align*}
    (U_1,\dots,U_t)&=N(U_1,\dots,U_t)M^g\\
    &=(NU_1M_1^g,\dots,NU_tM_t^g).
  \end{align*}
  It follows that $U_i=NU_iM_i^g$ which means that $M_i^g\in
  \stab{\langle M_i\rangle}{\Uvs_i}$, i.e., $r_i\mid g$ for any $i\in \{1,\dots,t\}$.
\item $g\mid \mathrm{lcm}(s_i\mid i\in \{1,\dots,t\})$ since, by
  definition of the $s_i$, $M^{\mathrm{lcm}(s_i\mid i\in
    \{1,\dots,t\})}\in \stab{G_M}{\Uvs}$.
\end{enumerate}

 \end{proof}

\begin{prop}\label{l:second}
  Let $k_i\leq n_i$, $U_i\in \F_q^{k_i\times d_i}$ be matrices
  with full rank, $\Uvs:=\rs(U)\in \Gr$ where
  $U=\diag(U_1,\dots,U_t)$, $\Cvs=\Uvs\langle M\rangle$ and
  $\Cvs_i:=\rs(U_i)\langle M_i\rangle$. It holds that
 \[|\Cvs|=\mathrm{lcm}(|\Cvs_i|\mid i\in \{1,\dots, t\})\] and
 \[d(\Cvs)\geq \min_{i\in \{1,\dots,t\}}\left\{d(\Cvs_i)\right\}.\]
  Moreover, we can distinguish the following cases:
  \begin{itemize}
  \item If $\gcd(|\Cvs_i|,|\Cvs_j|)=1$ for all $i\neq j$, then 
    \[
    d(\Cvs)= \min_{i\in \{1,\dots,t\}}\left\{d(\Cvs_i)\right\} .
    \]
  \item If $J=\{i\in \{1,\dots,t\}\mid |\Cvs_i|=|\Cvs|\}\neq
    \emptyset$, then \[d(\Cvs)\geq \sum_{j\in \{1,\dots, t\}\setminus
      J}d(\Cvs_j).\]
  \end{itemize}
  
\end{prop}

\begin{proof}
\begin{enumerate}
\item 
 Let us first derive the cardinality of $\Cvs$. Let $j:=\min\{i\in
  \N \mid \Uvs=\Uvs M^i\}$. Then
  \begin{eqnarray*}
    \rank\mat{U \\
      UM^j}& =& 
    \rank\mat{\diag(U_1,\dots,U_t) \\ 
      \diag(U_1M_1^j,\dots,U_tM_t^j)}\\
    &=&\sum_{i=1}^t \rank\mat{U_i\\U_iM_i^j}=k.
  \end{eqnarray*}
  Since $\rank\mat{U_i\\U_iM_i^j}\geq k_i$ for all
  $i\in\{1,\dots, t\}$, it follows that
  $\rs(U_i)=\rs(U_i)M^j$, implying that $|\Cvs_i|$ divides $j$ for
  all $i\in \{1,\dots,t\}$. By minimality we obtain the formula of the cardinality.

\item
  To show the bound on the minimum distance 
  assume without loss of generality that $d(\Cvs_1)\leq d(\Cvs_i)$ for
  $i\in \{1,\dots,t \}$ and $d_1\in \N$ such that 
  \[d_1=\max_{1\leq j<|\Cvs_1|}\left\{\dim (\rs(U_1)\cap\rs(U_1)M_1^j)\right\}.\] 
  For $1\leq j<|\Cvs|$, it holds that
  \begin{align*}
    \rank\mat{U\\UM^j}&=\sum_{i=1}^t
    \rank\mat{U_i\\U_iM_i^j} \\ 
    & \geq 2k_1-d_1+\sum_{i=2}^tk_i=k+k_1-d_1.
  \end{align*} 
  It follows that 
  \begin{align*}
    d(\Cvs)&=2 \min_{1\leq j< |\Cvs|}\left\{\rank\mat{U\\UM^j}\right\}-2k\\
    & \geq  2(k+k_1-d_1)-2k=2k_1-2d_1=d(\Cvs_1).
  \end{align*}
\item
  The inequality becomes an equality if $\gcd(|\Cvs_i|,|\Cvs_j|)=1$
  for any $i\neq j$, since there exists a $1\leq g<|\Cvs|$ such that
  \[g\equiv g_1 \pmod{|\Cvs_1|} \mbox{ and } g\equiv 0
  \pmod{|\Cvs_j|}\] for $j>1$. It follows that $d(\Uvs,\Uvs M^g)=d(\Cvs_1)$.
\item
  If instead $J:=\{i\in \{1,\dots,t\}\mid |\Cvs_i|=|\Cvs|\}$ is
  non-empty, then, for any $1\leq j < |\Cvs|$,  it holds that
  \[\rank\mat{U\\UM^j}\geq\sum_{i\in J}k_i+\sum_{i\notin
    J}(2k_i-d_i)=k+\sum_{i\notin J}(k_i-d_i)\] where $d_i=\max_{1\leq j<
    |\Cvs_i|}\left\{\dim(\rs(\bar{U}_i)\cap\rs(\bar{U}_i)M_i^j)\right\}$. This implies that
  \[d(\Cvs)\geq \sum_{j\notin J}d(\Cvs_j).\]

  \end{enumerate}
\end{proof}

\begin{prop}\label{l:first}
  Let $\Uvs:=\rs (U_1,\dots,U_t)\in \Gr$ where $k\leq n_i$ and $U_i\in
  \F_q^{k\times n_i}$ are matrices having full rank.
  Define 
  \begin{align*}
    \Cvs:=\Uvs\langle M\rangle\subset \Gr \mbox{ and
    }\Cvs_i:=\rs(U_i)\langle M_i\rangle\subset \mathcal{G}_{q}(k,n_i).
  \end{align*} 
  It holds that 
  \[d(\Cvs) \geq \min_{i\in \{1,\dots, t\}}\left\{d(\Cvs_i)\right\}.\]
\end{prop}

\begin{proof}
  For any $g\in \{1,\dots,|\Cvs|-1\}$ the following properties hold:
  \begin{enumerate}
  \item If $v=(v_1,\dots,v_t)\in \Uvs\cap\Uvs M^g$ where $v_i\in
    \F_q^{n_i}$ then $v_i\in \rs(U_i)\cap \rs(U_i)M_i^g$ for $i\in
    \{1,\dots,t\}.$
  \item There exists $j\in \{1,\dots,t\}$ such that $\rs(U_j)\neq \rs(U_j)M_j^g$.
  \end{enumerate}
  It follows that
  \[
    \max_{g}\{\dim(\Uvs\cap \Uvs M^g)\} \leq \hspace{4.5cm}\]
    \[\max_{i}\!\left\{\!\max_g \!\left\{\!\dim(\rs(U_i)\cap
        \rs(U_i)M_i^g)\mid \rs(U_i)\neq
        \rs(U_i)M_i^g\right\}\!\!\right\}\] 
    which implies that $d(\Cvs)\geq \min_{i\in \{1,\dots,t\}}
    \{d(\Cvs_i) \}$.

\end{proof}

\begin{ex} {\ }
  Let $p_1=x^4+x+1,p_2=x^6+x+1\in \F_2[x]$ and $M_1,M_2$ the respective companion matrices of $p_1,p_2$.
      Let
    \[
      U_1=\mat{1&0&0&0\\0&1&1&0}, \mbox{ }
      U_2=\mat{1&0&0&0&0&0\\0&1&0&1&1&1}
    \] and $\Cvs_1=\rs(U_1)\langle M_1\rangle$ and
    $\Cvs_1=\rs(U_2)\langle M_2\rangle$. It holds that $|\Cvs_1|=5$,
    $|\Cvs_2|=21$ and $d(\Cvs_1)=d(\Cvs_2)=4$.

\begin{enumerate}
  \item \emph{Proposition \ref{l:second}:}\\
      Let $\Uvs=\rs\mat{U_1 & 0 \\0 &U_2}\in \mathcal{G}_2(4,10)$ and
    $M=\diag(M_1,M_2)$,  then
    $|\Cvs|=105$ and $d(\Cvs)=4$.
 
  \item \emph{Proposition \ref{l:first}:}\\
    Let $\Uvs=\rs\mat{U_1 & U_2}$, $M=\diag(M_1,M_2)$ and
    $\Cvs=\Uvs\langle M\rangle$, then
      \[|\Cvs|=315 \mbox{ and } d(\Cvs)=4.\]

%
    
  \end{enumerate}
  
\end{ex}

\section{Cardinality and Minimum Distance of Cyclic Orbit Codes}\label{sec:cardinality}

In the previous section we investigated the cyclic subgroups of $\GL_n$ to classify the respective orbit codes. To compute the minimum distance and cardinality of a given cyclic orbit code we also need to take the starting point of that orbit into account. In this section we will show how to compute these parameters using the polynomial extension field representation of the elements of that starting point.

We will explain in detail the case of irreducible cyclic orbit codes and give some remarks in the end how to generalize this to arbitrary cyclic subgroups of $\GL_n$.


\subsection{Irreducible cyclic orbit codes}\label{icoc}

\begin{defi}
\begin{enumerate}
\item A matrix $A\in \GL_n$ is called \emph{irreducible} if $\F_q^n$ contains no non-trivial $A$-invariant subspace, otherwise it is called \emph{reducible}.
\item A non-trivial subgroup $G\le \GL_n$ is called \emph{irreducible} if $\F_q^n$ contains no non-trivial $G$-invariant subspace, otherwise it is called \emph{reducible}.
\end{enumerate}
\end{defi}

A cyclic group is irreducible if and only if its generator matrix is irreducible. Moreover, an invertible matrix is irreducible if and only if its characteristic polynomial is irreducible. The rational canonical form of such an invertible matrix is the companion matrix of its characteristic polynomial.

It follows that any cyclic irreducible subgroup of $\GL_n$ is conjugate to a group generated by a companion matrix of an irreducible polynomial.
Moreover, we know from Theorem \ref{t:rep_grp}  that groups generated by companion matrices of an irreducible polynomial of the same order are conjugate and from Lemma \ref{thm:isometricorbitcodes} that orbit codes defined by conjugate groups are isometric.

Therefore it is sufficient to
characterize the orbits of cyclic groups generated by companion
matrices of irreducible polynomials of degree $n$ of different orders.

We first need the following definitions and results on irreducible
polynomials over finite fields.

Let $p(x)=p_0+p_1x+\ldots+p_{n-1}x^{n-1}+x^n$ be a monic irreducible
polynomial of degree $n$ over the finite field $\F_{q}$. 

In the finite field $\F_{q}[x]/(p(x))$, the modular multiplication of
an element $v(x)=v_0+v_1x+\ldots+v_{n-1}x^{n-1}$ by the fixed
polynomial $x$ yields
\begin{align*}
w(x) &=v(x)\cdot x \pmod{p(x)} \\
 &= -v_{n-1}p_0+\sum_{i=1}^{n-1}(v_{i-1}-v_{n-1}p_{i})x^{i}.
\end{align*}
Using a vector representation of $\F_{q}[x]/(p(x))$, the modular
multiplication by $x$, $w(x)=v(x)\cdot x\pmod{p(x)}$ yields a linear
mapping
\[
(w_0,w_1,\ldots,w_{n-1})=(v_0,v_1,\ldots,v_{n-1}) M_p
\]
where $M_p$ is the companion matrix of $p(x)$ as defined in Definition
\ref{d:com_mat}. 

If $p(x)$ is a \emph{primitive} polynomial,
i.\,e. the minimal polynomial of a multiplicative generator $\alpha$
of $\F_{q^n}^*:=\F_{q^n} \setminus \{0\}$, the group generated by
$M_p$ has order $q^n-1$ and is known as the \emph{Singer group}. This
notation is used e.\,g. by Kohnert et al. in their network code
construction \cite{el10, ko08p}. Elsewhere $M_p$ is called
\emph{Singer cycle} or \emph{cyclic projectivity} (e.\,g. in
\cite{hi98}).

If we substitute the indeterminate $x$ in the polynomials $v\in
\F_{q}$ by the generator $\alpha\in \mathbb{F}_{q^n}^*$, the modular
multiplication with $x$ corresponds to the multiplication with
$\alpha$ and hence to multiplication with the companion matrix $M_p$:

\begin{lem}\label{thm:compmult}
  Let $p(x)$ be an irreducible polynomial over $\mathbb{F}_{q}$ of
  degree $n$ and $P$ its companion matrix.  Furthermore let $\alpha\in
  \mathbb{F}_{q^n}^*$ be a root of $p(x)$ and $\phi$ be the canonical
  homomorphism
\begin{align*}
\phi: \mathbb{F}_q^n &\longrightarrow \mathbb{F}_{q^n}\\ 
(v_0,\dots,v_{n-1})&\longmapsto \sum_{i=0}^{n-1} v_i \alpha^{i}.
\end{align*}
Then the multiplication with $M_p$ resp. $\alpha$ commutes with the
mapping $\phi$, i.\,e. for all $v\in \F_q^n$ we get
\[
\phi(vM_p)=\phi(v)\alpha .
\]
\end{lem}


\begin{defi}
A \emph{multiset} is a generalization of the notion of set in which members are allowed to appear more than once. To distinguish it from usual sets $\{x \in X\}$ we will denote multisets by $\{\{x \in X\}\}$. The number of times an element $x$ belongs to the multiset $X$ is the \emph{multiplicity} of that element, denoted by $m_X(x)$.
\end{defi}

We will first investigate the simpler case of primitive companion
matrices and then generalize it to arbitrary irreducible cyclic orbit
codes.


\subsubsection*{Primitive Generator}\label{prim}

For this part of the paper let $p(x)\in \F_{q}[x]$ always be a
primitive polynomial of degree $n$ (i.\,e. the order of $p(x)$ is
$q^{n}-1$) and $\alpha$ be a root of it. Thus $\alpha$ is a primitive
element of $\F_{q^n}$. It follows that for any non-zero element $u\in \F_{q}^{n}$ there exists an $i\in \Z_{q^{n}-1}$ such that $\phi(u) = \alpha^{i}$.

The following fact is a generalization of Lemma 1 from \cite{ko08p} and has been formulated in a similar manner in \cite{tr11}.

\begin{thm}\label{thmprim}
Assume $\mathcal{U}=\{0,u_1,\dots,u_{q^k-1}\} \in \Gr$, where $u_{i} \in \F_{q}^{n}\setminus \{0\}$ and let $b_i \in \Z_{q^n-1}$ such that
  \[\phi(u_i)=\alpha^{b_i} \hspace{0.7cm} \forall i=1,\dots,q^k-1 .\]
  Let $d$ be minimal such that any element of the set
  \[D:=\{\{b_{m}- b_{l} \mod q^n-1 \mid l,m \in \mathbb{Z}_{q^k-1}, l\neq
  m\}\}\] 
has multiplicity less than or equal to $q^d-1$, i.\,e. a quotient of two
  elements in the field representation appears at most $q^d-1$ times
  in the set of all pairwise quotients. If $d<k$ then the orbit of the group
  generated by the companion matrix $P$ of $p(x)$ on $\mathcal{U}$ is
  an orbit code of cardinality $q^{n}-1$ and minimum distance $2k-2d$.
\end{thm}
\begin{proof}
  In field representation the elements of the orbit code are:
  \begin{align*}
  C_{0}=&\{\alpha^{b_{1}}, \alpha^{b_{2}},...,\alpha^{b_{q^{k}-1}}\}\cup \{0\}\\
  C_{1}=&\{\alpha^{b_{1}+1}, \alpha^{b_{2}+1},...,\alpha^{b_{q^{k}-1}+1}\}\cup \{0\}\\
  \vdots &\\
  C_{q^{n}-2}=&\{\alpha^{b_{1}+q^{n}-2}, ...,\alpha^{b_{q^{k}-1}+q^{n}-2}\}\cup \{0\}
  \end{align*}
First we compute the minimum distance of the code. 
We know from Theorem \ref{thm:basic} that it suffices to compute the minimum distance and therefor the intersection of $C_{0}$ with $C_{h}$ for $h=1, \dots, q^{n}-2$.  A non-zero element $\alpha^{b_{i}} \in C_0$ is also in $C_{h}$ if and only if
\[
\alpha^{b_{i}} = \alpha^{b_j+h} \quad \textnormal{ for some } j\in \{1,\dots, q^{k}-1\} \]
\[\iff b_{i}-b_{j} \equiv h \mod q^{n}-1  .\]
But by assumption there are at most $q^{d}-1$ many pairs $(b_i, b_j)$ fulfilling this equation. Thus,
\[\dim (C_{0}\cap C_{h}) \leq d \quad \forall h\in \{1,\dots, q^{n}-2\}   .\]
Since we chose $d$ minimal, this inequality is actually an equality for some $h$, hence the minimum distance of the code is $2k-2d$.

The cardinality of the code follows from the fact that $d<k$, which implies that all elements of the orbit are distinct.
\end{proof}

\begin{prop}\label{d=k}
In the setting of before, if $d=k$, one gets orbit elements with full
intersection which means they are the same vector space. 
\begin{enumerate}
\item Let $D':=D\setminus \{\{a\in D \mid m_D(a)=q^{k}-1\}\}$ and $d':=\log_{q}(\max\{m_{D'}(a) \mid a\in D'\}+1)$. Then the minimum distance of the code is $2k - 2d'$. 
\item Let $m$ be the least element of $D$ of multiplicity $q^{k}-1$. Then the cardinality of the code is $m-1$.
\end{enumerate}
\end{prop}
\begin{proof}
\begin{enumerate}
\item Since the minimum distance of the code is only taken between distinct vector spaces, one has to consider the largest intersection of two elements whose dimension is less than $k$.
\item Since 
\[\Uvs P^{m} = \Uvs \implies \Uvs P^{lm} = \Uvs \quad \forall l \in \N\]
and the elements of $D$ are taken modulo the order of $P$, one has to choose the minimal element of the multiset $\{\{a \in D \mid m_D(a)=q^{k}-1\}\}$ for the number of distinct vector spaces in the orbit.
\end{enumerate}
\end{proof}

Note, that for $q>2$ it holds that  $\{P^{i}\mid i=0,\dots, \ord(P)\}=\F_{q}[P]$ and thus contains all scalar multiples of each power of $P$. Therefore, there are at most $(q^{n}-1)/(q-1)$ different vector spaces in the orbit and one will always get $d=k$ when applying Theorem \ref{thmprim}.

\begin{algorithm}
\caption{Computing cardinality and minimum distance of primitive cyclic orbit codes in $\Gr$}
\label{alg1}                       
\begin{algorithmic}       
\REQUIRE $p(x)\in \F_{q}[x]$ primitive of degree $n$, $\alpha$ a root of $p(x)$ and $\Uvs=\{0,u_{1},..., u_{q^{k}-1}\} \in \Gr$
\FOR{$u_{i}$ in $\Uvs\setminus\{0\}$}
\STATE store $\bar{u}_{i}:=\log_{\alpha}(\phi(v))$
\ENDFOR
\STATE store all differences $\bar{u}_i-\bar{u}_j$ in the difference multiset $D$

\STATE set $c:= \max\{m_D(a) \mid a \in D\}$
\IF{$\log_{q}(c+1)\neq k$}
\RETURN $q^{n}-1, 2(k-\log_{q}(c+1))$
\ELSE
\STATE set $c':=\max\{m_D(a) \mid a \in D $ and $ m_D(a)<c\}$
\STATE set $m:= \min\{a \in D \mid m_D(a)=c\}$
\RETURN $m-1, 2(k-\log_{q}(c'+1))$
\ENDIF
\end{algorithmic}
\end{algorithm}

Tables \ref{table1}--\ref{table3} list binary codes found by randomly picking points in the Grassmannian as initial points of the orbit and computing the minimum distance. Since this search was not exhaustive, the blank fields in the tables do not imply that there are no codes for this set of parameters, but that we have not found any.

\begin{table}[h]
\begin{center}
\begin{tabular}{*{6}{|c}|c|c|}
\hline
\backslashbox{$d(\Cvs)$}{$n$}&4&5&6&7&8&9&10 \\\hline
2&   15&31&63&127&255&511&1023 \\\hline
4&   5&-&21&-&85&-&341 \\\hline
\end{tabular}
\end{center}
\caption{Cardinality of binary primitive cyclic orbit codes for $k=2$.}
\label{table1}
\end{table}

\begin{table}[h]
\begin{center}
\begin{tabular}{*{6}{|c}|c|c|}
\hline
\backslashbox{$d(\Cvs)$}{$n$}&6&7&8&9&10&11&12 \\\hline
2&   63&127&255&511&1023&2047&4095 \\\hline
4&   63&127&255&511&1023&2047&4095 \\\hline
6&   9&-&-&73&-&-&585 \\\hline
\end{tabular}
\end{center}
\caption{Cardinality of binary prim. cyclic orbit codes for $k=3$.}
\end{table}

\begin{table}[h]
\begin{center}
\begin{tabular}{*{6}{|c}|c|c|}
\hline
\backslashbox{$d(\Cvs)$}{$n$}&8&9&10&11&12 \\\hline
2&   255&511&1023&2047&4095 \\\hline
4&   255&511&1023&2047&4095 \\\hline
6&   -&511&1023&2047&4095 \\\hline
8&   17&-&-&-&273 \\\hline
\end{tabular}
\end{center}
\caption{Cardinality of binary prim. cyclic orbit codes for $k=4$.}
\label{table3}
\end{table}

In the following we will show that spread codes can be constructed as cyclic orbit codes (cf. \cite{tr11}).
For this let $\alpha$ be a primitive element of $\F_{q^n}$ and assume $k|n$ and
$c:=\frac{q^n-1}{q^k-1}$. Naturally, the subfield $\F_{q^{k}}\leq \F_{q^{n}}$ is also an $\F_{q}$-subspace of $\F_{q^{n}}$. On the other hand, $\F_{q^{k}} = \{\alpha^{ic} \mid i=0,...,q^{k}-2\}\cup\{0\}$.

\begin{lem}
  For every $\beta\in\F_{q^n}$ the set
$$
\beta\cdot\F_{q^k}=\{\beta\alpha^{ic} \mid i=0,...,q^{k}-2\}\cup\{0\}
$$
defines an $\F_q$-subspace of dimension $k$.
\end{lem}
\begin{proof}
Since $\F_{q^{k}}$ is a subspace of dimension $k$ and
\begin{align*}
\varphi_\beta:\ \F_{q^n} & \longrightarrow \F_{q^n}\\ 
u &\longmapsto \beta u
\end{align*}
is an $\F_q$-linear isomorphism, it follows that 
$\varphi_\beta(\F_{q^k})=\beta\cdot\F_{q^k}$ is an $\F_q$-subspace of dimension $k$. 
\end{proof}

\begin{thm}
 The set
$$
\mathcal{S}=\left\{ \alpha^i\cdot\F_{q^k}\mid i=0,\ldots,c-1\right\}
$$
is a spread of $\F_{q^{n}}$ and thus defines a spread code in $\Gr$. 
\end{thm}
\begin{proof}
  By a simple counting argument it is enough to show that the subspaces
  $\alpha^i\cdot\F_{q^k}$ and $\alpha^j\cdot\F_{q^k}$ have only
  trivial intersection whenever $0\leq i<j\leq c-1$.  For this assume
  that there are field elements $\beta_i,\beta_j\in \F_{q^k}$, such that
$$
v=\alpha^i \beta_i=\alpha^j \beta_j\in \alpha^i\cdot\F_{q^k}\cap
\alpha^j\cdot\F_{q^k}.
$$
If $v\neq 0$ then $\alpha^{i-j}=\beta_j \beta_i^{-1}\in \F_{q^k}.$ But this
means $i-j\equiv 0 \mod c$ and $ \alpha^i\cdot\F_{q^k}=
\alpha^j\cdot\F_{q^k}$, which contradicts the assumption.  It follows that $\mathcal{S}$ is a spread.
\end{proof}

We now translate this result into a matrix setting.
For this let $\phi$  denote the canonical homomorphism as
defined in Theorem \ref{thm:compmult}.

\begin{co}\label{spread}
  Assume $k|n$.  Then there is a subspace 
 $\mathcal{U} \in \Gr$ such that the cyclic orbit code obtained by
  the  group action of a 
 a primitive companion matrix is a code with
  minimum distance $2k$ and cardinality $\frac{q^n-1}{q^k-1}$. Hence
  this  irreducible cyclic orbit code is a spread code.
\end{co}
\begin{proof}
  In the setting of the previous theorem represent $\F_{q^k}\leq \F_{q^n}$ as the
  row space of a $k\times n$ matrix $U$ over $\F_q$ and, using the
  same basis over $\F_q$, represent the primitive element $\alpha$
  with its respective companion matrix $P$. Then the orbit code $\Cvs = \rs(U) \langle
  P\rangle$ has all the desired properties.  
\end{proof}

\begin{ex}\cite{tr11}
  Over the binary field let $p(x):=x^6+x+1$ be primitive, $\alpha$ a root
  of $p(x)$ and $P$ its companion matrix.
  \begin{enumerate}
  \item For the 3-dimensional spread compute $c=\frac{63}{7}=9$ and
    construct a basis for the starting point of the orbit:
    \begin{align*}
      u_1&=\phi^{-1} (\alpha^0)=\phi^{-1} (1)=(100000)\\
      u_2&=\phi^{-1} (\alpha^c)=\phi^{-1}(\alpha^9)=\phi^{-1}(\alpha^4+\alpha^3)=(000110)\\
      u_3&=\phi^{-1} (\alpha^{2c})=\phi^{-1}(\alpha^{18})=\phi^{-1}(
      \alpha^3+\alpha^2+\alpha+1)\\
	  &= (111100)
    \end{align*}
    The starting point is
    \[ \mathcal{U}=\rs\left[\begin{array}{cccccc} 1&0&0&0&0&0\\
        0&0&0&1&1&0\\ 1&1&1&1&0&0 \end{array}\right] 
\] 
and the orbit of the group
    generated by $P$ on $\mathcal{U}$ is a spread code.
  \item For the 2-dimensional spread compute $c=\frac{63}{3}=21$ and
    construct the starting point
    \begin{align*}
      u_1&=\phi^{-1} (\alpha^0)=\phi^{-1}(1)=(100000)\\
      u_2&=\phi^{-1} (\alpha^c)=\phi^{-1}(\alpha^{21})=\phi^{-1}(\alpha^2+\alpha+1)\\
	  &=  (111000)
    \end{align*}
    The starting point is
    \[\mathcal{U}=\rs\left[\begin{array}{cccccc} 1&0&0&0&0&0\\
        1&1&1&0&0&0 \end{array}\right]
\]
    and the orbit of the group generated by $P$ is a spread code.
  \end{enumerate}
\end{ex}

\subsubsection*{Non-Primitive Generator} 

For this subsection let $p(x)\in \F_{q}[x]$ be irreducible but not primitive, and, as before, $\alpha$ a root of it and $P$ its companion matrix.

\begin{prop}\label{non1}
  Let $G=\langle P \rangle$ the group generated by $P$. If $\mathcal{U} \in \Gr$ such that 
  \[v\neq w \implies vG \neq wG  \quad \forall \: v,w \in \mathcal{U} , \]
  then $\mathcal{U}G$ is an orbit code with minimum distance $2k$ and
  cardinality $\mathrm{ord}(P)$.
\end{prop}
\begin{proof}
The cardinality follows from the fact that each element of $\mathcal{U}$ has its own orbit of
  cardinality $\mathrm{ord}(P)$. Moreover, no code words intersect
  non-trivially, hence the minimum distance is $2k$.   
\end{proof}

Note that, if the order of $P$ is equal to $\frac{q^{n}-1}{q^{k}-1}$,
these codes are again spread codes.

\begin{ex}
  Over the binary field let $p(x)=x^4+x^3+x^2+x+1$, $\alpha$ a root of
  $p(x)$ and $P$ its companion matrix. Then $\F_{2^4}\setminus \{0\}$
  is partitioned into
  \[\{\alpha^i | i=0, \dots , 4\} \cup \{\alpha^i(\alpha+1) | i=0,
  \dots , 4\} \]
  \[\cup \{\alpha^i(\alpha^2+1) | i=0, \dots , 4\} .\] Choose
  \begin{align*}
    u_1=&\phi^{-1}(1)=\phi^{-1}(\alpha^0)=(1000)\\
    u_2=&\phi^{-1}(\alpha^3+\alpha^2)=\phi^{-1}(\alpha^2(\alpha+1))=(0011)\\
    u_3=&u_1+u_2=\phi^{-1}(\alpha^3+\alpha^2+1)
    =\phi^{-1}(\alpha^4(\alpha^2+1))\\=&(1011)
  \end{align*}
  such that each $u_i$ is in a different orbit of $\langle P\rangle$
  and $\mathcal{U}=\{0,u_1,u_2,u_3\}$ is a vector space.
  Then the orbit of $\langle P\rangle$ on $\mathcal{U}$ has minimum
  distance $4$ and cardinality $5$, hence it is a spread code.
\end{ex}

The following theorem generalizes this result to any possible starting point in $\Gr$ and can be proven analogously to Theorem \ref{thmprim}.

\begin{thm} \cite{tr11}
  Let $P, G, \mathcal{U}$ be
  as before and $O_1,...,O_{\ell}$ the distinct orbits of $G$ on $\F_{q}^{n}\setminus \{0\}$. Assume that $m_i$ elements of $\mathcal{U}$ are in the
  same orbit $O_i$ ($i=1,\dots, \ell$). Apply the theory of Section
  \ref{prim} to each orbit $O_i$ and find the corresponding $d_i$ from
  Theorem \ref{thmprim}.  Then the following cases can occur:
  \begin{enumerate}
  \item No intersections of two different orbits coincide. Define $d:=
    \max_{i}d_i$.  Then the orbit of $G$ on $\mathcal{U}$
    is a code of cardinality $\mathrm{ord}(P)$ and minimum distance
    $2k-2d$.
  \item Some intersections coincide among some orbits. Then the corresponding $d_i$'s
    add up and the maximum of these is the maximal intersection number $d$. 
  \end{enumerate}
Mathematically formulated: Assume the orbits are of the type
\[O_{i}= \{\tilde{p}_i(\alpha) \alpha^{j} \mid j=1,\dots, \ord(P)\} \quad \forall \: i=1,\dots,\ell\]
for some fixed $\tilde{p}_i(\alpha)\in \F_{q}[\alpha]$. Then for any $u_{j} \in O_i$ there exists $b_{(i,j)} \in \Z_{\ord(P)}$ such that
\[\phi(u_{j}) = \tilde{p}_i(\alpha) \alpha^{b_{(i,j)}} .\]
For $i=1,\dots,\ell$ define 
\[a_{(i,\mu,\lambda)} := b_{(i,\mu)} - b_{(i,\lambda)}\]
and the difference multisets
\[D_i:=\{\{a_{(i,\mu,\lambda)}\mid \mu,\lambda \in \{1,\dots, \mathrm{ord}(P)-1\}\}\},\]
\[D:=\bigcup_{i=1}^{\ell} D_i .\]
Let $d := \log_q(\max \{m_D(a) \mid a\in D\}+1)$.
If $d<k$, then the orbit of $G$ on $\mathcal{U}$ is
  a code of cardinality $\mathrm{ord}(P)$ and minimum distance $2k-2d$.
\end{thm}

If $d=k$, Proposition \ref{d=k} also holds in this (non-primitive) case.
The algorithm works analogously to Algorithm \ref{alg1}.


Tables \ref{table4}--\ref{table6} list binary cyclic orbit codes generated by irreducible polynomials of degree $10$ with different orders, found by random search. By ``X'' we denote a set of parameters that cannot be fulfilled. As before, a ``$-$'' means that we have not found a code for these parameters and not neccessarily that there do not exist any.

\begin{table}[h]
\begin{center}
\begin{tabular}{*{6}{|c}|c|c|}
\hline
\backslashbox{$d(\Cvs)$}{$k$}&2&3&4&5 \\\hline
2&   33& 33& 33&33 \\\hline
4&   33& 33& 33&33 \\\hline
6&   X& 33& 33& 33\\\hline
8&   X& X& -&33\\\hline
10&  X& X& X&-\\\hline
\end{tabular}
\end{center}
\caption{Cardinality of codes generated by a polynomial of order $33$ for $n=10$.}
\label{table4}
\end{table}

\begin{table}[h]
\begin{center}
\begin{tabular}{*{6}{|c}|c|c|}
\hline
\backslashbox{$d(\Cvs)$}{$k$}&2&3&4&5 \\\hline
2&   93& 93& 93& 93\\\hline
4&   93& 93& 93 & 93 \\\hline
6&   X& 93& 93& 93 \\\hline
8&   X& X& -& -  \\\hline
10&   X& X& X& -\\\hline
\end{tabular}
\end{center}
\caption{Cardinality of codes generated by a polynomial of order $93$ for $n=10$.}
\end{table}

\begin{table}[h]
\begin{center}
\begin{tabular}{*{6}{|c}|c|c|}
\hline
\backslashbox{$d(\Cvs)$}{$k$}&2&3&4&5 \\\hline
2&   341& 341& 341& 341 \\\hline
4&   341& 341& 341& 341 \\\hline
6&   X& -& 341& 341 \\\hline
8&   X& X& -& -  \\\hline
10&   X& X& X& -\\\hline
\end{tabular}
\end{center}
\caption{Cardinality of codes generated by a polynomial of order $341$ for $n=10$.}
\label{table6}
\end{table}


\subsection{Completely reducible cyclic orbit codes}

We call an orbit code completely reducible if its generating group is completely reducible. In general, a group is \emph{completely reducible} or \emph{semisimple} if it is the direct product of irreducible groups. For the $\GL_n$-action on $\F_q^n$ a subgroup $H\leq \GL_n$ is compeletely reducible if $\F_q^n$ is the direct sum of subspaces $V_1,\dots, V_i$ which are $H$-invariant but do not have any $H$-invariant proper subspaces.

In the cyclic case these groups are exactly the ones where the blocks of the rational canonical form of the generator matrix are companion matrices of irreducible polynomials, i.e. all the elementary divisors have exponent $1$.
Because of this property one can use the theory of irreducible cyclic orbit codes block-wise to compute the minimum distances of the block component codes and hence the minimum distance of the whole code.

For simplicity we will explain how the theory from before generalizes in the case of generator matrices whose RCF has two blocks that are companion matrices of primitive polynomials. The generalization to an arbitrary number of blocks and general irreducible polynomials is then straight-forward.

Assume our generator matrix $P$ is of the type
\[P=\left( \begin{array}{cc} P_1 & 0 \\ 0 & P_2
           \end{array}
\right)\]
where $P_1, P_2$ are companion matrices of the primitive polynomials $p_1(x), p_2(x)\in \F_{q}[x]$ with $\deg(p_1)=n_1, \deg(p_2)=n_2$, respectively. Furthermore let
\[U=\left[ \begin{array}{cc}U_1 &U_2    \end{array}
\right]\]
be the matrix representation of the starting point $\Uvs \in \Gr$ such that $U_1 \in Mat_{k\times n_1}, U_2\in Mat_{k\times n_2}$.
Then
\[\Uvs P^i = \rs \left[\begin{array}{cc} U_1 P_1^i &   U_2 P_2^i \end{array}\right] .\]
By $\phi^{(n_{1})}: \F_{q}^{n_{1}}\rightarrow \F_{q^{n_{1}}}$ and $ \phi^{(n_{2})} : \F_{q}^{n_{2}}\rightarrow \F_{q^{n_{2}}}$ we denote the standard vector space isomorphisms.

The algorithm for computing the minimum distance of the orbit code is analogous to before, but it is now set in $\F_{q^{n_{1}}}\times \F_{q^{n_{2}}}$.

\begin{thm}
Let $\alpha_1, \alpha_2$ be primitive elements of $\F_{q^{n_1}}, \F_{q^{n_2}}$ respectively. 
\begin{align*}
\phi^{(n_{1},n_{2})}: \F_q^n &\longrightarrow \F_{q^{n_1}} \times \F_{q^{n_2}}\\
(u_1,...,u_n) &\longmapsto (\phi^{(n_1)}(u_1,...,u_{n_1}), \phi^{(n_2)}(u_{n_1+1},...,u_n))
\end{align*}
is a vector space isomorphism. Moreover, $u = v P^{i}$ for some $u,v \in \F_{q}^{n}$ if and only if 
\begin{enumerate}
\item $\phi^{(n_{1})} (u_1,...,u_{n_1}) = \phi^{(n_{1})} (v_1,...,v_{n_1}) \alpha_{1}^{i} $ and
\item $\phi^{(n_2)}(u_{n_1+1},...,u_n) = \phi^{(n_2)}(v_{n_1+1},...,v_n) \alpha_{2}^{i} $.
\end{enumerate}
\end{thm}

\begin{proof}
$\phi^{(n_{1},n_{2})}$ is a vector space isomorphism because $\phi^{(n_{1})}$ and $\phi^{(n_{2})}$ are. The second statement follows since
\begin{align*}
& u = v P^{i}\\ \iff & \phi^{(n_{1},n_{2})}(u) = \phi^{(n_{1}, n_{2})}(vP^{i}) \\
\iff & \phi^{(n_{1})} ((u_1,...,u_{n_1})) = \phi^{(n_{1})} ((v_1,...,v_{n_1})P_{1}^{i})  \;\wedge\\ 
& \phi^{(n_2)}((u_{n_1+1},...,u_n)) = \phi^{(n_2)}((v_{n_1+1},...,v_n) P_{2}^{i})  .
\end{align*}
\end{proof}

Thus, if $\phi^{(n_1)}(u_i)\neq 0$ and $\phi^{(n_2)}(u_i)\neq 0$ for all non-zero elements $u_i$ of a given vector space $\Uvs \in \Gr$, in the algorithm we have to create the difference set of all $2$-tuples corresponding to the powers of $\alpha_{1}$ and $\alpha_{2}$ and proceed as usual.

\begin{prop}\label{prop30}
Assume $\mathcal{U}=\{0,u_1,\dots,u_{q^k-1}\} \in \Gr$, and for all $u_i$ there exist $b_i, b'_i$ such that
  \[\phi^{(n_1,n_2)}(u_i)=(\alpha_1^{b_i}, \alpha_2^{b'_i}) \hspace{0.7cm} \forall i=1,\dots,q^k-1 .\]
  Let $d$ be minimal such that any element of the multiset
  \begin{eqnarray*}
D:=\{\{(b_{m}- b_{\ell} \mod q^{n_1}-1 , b'_m-  b'_\ell \mod q^{n_2}-1)  \mid \\ \ell,m \in \mathbb{Z}_{q^k-1}, \ell\neq
  m\}\}   
  \end{eqnarray*}
has multiplicity less than or equal to $q^d-1$. If $d<k$ then the orbit of the group
  generated by $P$ on $\mathcal{U}$ is
  an orbit code of cardinality $ \ord (P) = \mathrm{lcm}(q^{n_1}-1, q^{n_2}-1)$ and minimum distance $2k-2d$.
\end{prop}

Since it is possible that $u=(u_1,\dots,u_n)\in \F_q^n$ is non-zero but all $u_1=\dots=u_{n_1}=0$ (or the second part of the coefficients), we have to take the zero element into account when counting intersection elements:

\begin{thm}
Assume $\mathcal{U}=\{0,u_1,\dots,u_{q^k-1}\} \in \Gr$, and for all $u_i$ either
\begin{enumerate}
 \item   $\phi^{(n_1,n_2)}(u_i)=(\alpha_1^{b_i}, \alpha_2^{b'_i})  $ ,
 \item   $\phi^{(n_1,n_2)}(u_i)=(\alpha_1^{b_i}, 0)$ or
 \item  $\phi^{(n_1,n_2)}(u_i)=(0, \alpha_2^{b'_i})$.
\end{enumerate}
Denote by $S_1, S_2, S_3$ the sets of all elements of the first, second and third type, respectively, and construct the difference sets
  \begin{eqnarray*}
D_1:=\{\{(b_{m}- b_{\ell} \mod q^{n_1}-1 , b'_m-  b'_\ell \mod q^{n_2}-1)  \mid \\ u_{\ell},u_m \in S_1, \ell\neq m\}\} , \\
D_2:=\{\{(b_{m}- b_{\ell} \mod q^{n_1}-1, j)  \mid  u_{\ell},u_m \in S_2, \ell\neq m, \\j=1,\dots,q^{n_2}-1\}\} ,  \\
D_3:=\{\{(j, b'_m-  b'_\ell \mod q^{n_2}-1)  \mid  u_{\ell},u_m \in S_3, \ell\neq m, \\j=1,\dots,q^{n_1}-1\}\}   
  \end{eqnarray*}
and
\[D:=D_1 \cup D_2 \cup D_3 .
\]
  Let $d$ be minimal such that any element of $D$ has multiplicity less than or equal to $q^d-1$. If $d<k$ then the orbit of the group
  generated by $P$ on $\mathcal{U}$ is
  an orbit code of cardinality $ \ord (P) = \mathrm{lcm}(q^{n_1}-1, q^{n_2}-1)$ and minimum distance $2k-2d$.
\end{thm}

\begin{proof}
As in the irreducible case we want to count the number of intersecting elements and use the fact that $\langle P_1\rangle , \langle P_2 \rangle$ act transitively on $\F_q^{n_1}\setminus \{0\}$ and $\F_q^{n_2}\setminus \{0\}$ respectively. 
Let $\pi_1: \F_q^n \rightarrow \F_q^{n_1}, (u_1,\dots,u_n)\mapsto (u_1,\dots,u_{n_1})$ and $\pi_2: \F_q^n \rightarrow \F_q^{n_2}, (u_1,\dots,u_n)\mapsto (u_{n_1+1},\dots,u_{n_2})$. 
\begin{enumerate}
 \item 
Assume $u\in S_3$, i.e. $\pi_1(u)=0$. Then
\[\pi_1(uP^i)=\pi_1(u)P_1^i =0 \quad \forall\: i=1,\dots,\ord(P) .\]
Thus, $uP^j \neq v$ for all $v \in S_1\cup S_{2}$ and $j=1,\dots,\ord(P)$, i.e. intersection with $u$ can only happen inside $S_3$. 

On the other hand, if $\pi_2(u)=\pi_2(u)P_2^j$ for some $j$, then also $u=uP^j$, which is why the second entry of the tuple can run over all possible values.
 \item
For $u\in S_2$ the analogue holds.
 \item
For $u\in S_1$ we can use Proposition \ref{prop30}.
\end{enumerate}
Since we have to check if some of the intersections inside the sets $S_1,S_2,S_3$ occur at the same element of the orbit we have to count the intersection inside the union of the difference sets.
\end{proof}

Like in the irreducible case, if $d=k$, one gets orbit elements with full
intersection. Let $D':=D\setminus \{\{a=(a_1, a_2) \in D \mid m_D(a)=q^{k}-1\}\}$ and $d':=\log_{q}(\max\{m_{D'}(a) \mid a\in D'\}+1)$. Then the minimum distance of the code is $2k - 2d'$.
Moreover, let  $m:=\min\{\mathrm{lcm}(a_1,a_2) \mid a=(a_1,a_2)\in D, m_D(a)=q^k-1\}$. Then the cardinality of the code is $m-1$.


\subsection{Non- completely reducible cyclic orbit codes}

If a matrix has elementary divisors of exponent larger than one, the group generated by it is not completely reducible. In this subsection we will explain what happens in this case when you try to apply the theory from the previous subsections. For simplicity we will describe the case of polynomials that are squares of an irreducible one. This can easily be generalized to higher exponents. Along the lines of the previous subsection one can then translate the theory to more than one irreducible factor block-wise.

Let $p(x) \in \F_{q}[x]$ be irreducible of degree $n/2$ and $f(x)=p^{2}(x)$. Denote by $P_{f}$ the companion matrix of $f(x)$. Since a root of $f(x)$ is also a root of $p(x)$ we can not use it to represent $\F_{q}[x]_{<n}\cong \F_{q}^{n}$. Therefore we will now use polynomials in the variable $x$ and the standard vector space isomorphism $\phi : \F_{q}^{n} \rightarrow \F_{q}[x]_{<n} , (u_{1},\dots, u_{n}) \mapsto \sum_{i=1}^{n} u_{i}x^{i-1}$. Then the following still holds:
\[\phi (uP_{f}) = \phi(u) x  \mod f(x).\]

Hence, one can still translate the question of finding the intersection number into the polynomial setting by finding the respective $x^{i}$ that maps one element to another element of the initial point $\Uvs \in \Gr$. The difference to the cases before is that we do not have a field structure anymore, thus in general we cannot divide one element by the other modulo $f(x)$ to find the corresponding $x^{i}$. More precisely, we can divide by the units of $\F_{q}[x]/{(f(x))}$. In the other cases we can find the $x^{i}$ by brute force.

\begin{thm}
Assume $\mathcal{U}=\{0,u_1,\dots,u_{q^k-1}\} \in \Gr$, 
  \[\phi(u_i)=\phi(u_{j}) x^{b_{ij}} \]
  for all $\phi(u_{i}), \phi(u_{j})$ that lie on the same orbit of $\langle x \rangle$
  and $d$ be minimal such that any element of the multiset
  \[D:=\{\{b_{ij} \mod q^n-1 \mid i,j \in \mathbb{Z}_{q^k-1}, i\neq
  j\}\}\] has multiplicity less than or equal to $q^d-1$. If $d<k$ then the orbit of the group
  generated by $P_{f}$ on $\mathcal{U}$ is
  an orbit code of cardinality $q^{\frac{n}{2}}-1$ and minimum distance $2k-2d$.
\end{thm}

\section{Decoding primitive cyclic orbit codes}
\label{sec:decoding}


In this section we will show how to decode primitive cyclic orbit codes. Let $\Uvs \in \Gr$, $\alpha \in \F_{q}^{n}$ a primitive element, $P \in \GL_{n}$ the corresponding companion matrix, $G=\langle P\rangle$ and $\Cvs = \Uvs G$ the generated orbit code with minimum distance $2\delta$.

If $\Vvs \in \Cvs$ denotes the sent code word and $\Rvs \in \PP_{q}(n)$ the received vector space, then $\Rvs$ is uniquely decodable if $d(\Rvs, \Vvs)\leq \delta-1$. 

A minimum distance decoder finds 
$A\in G$, such that
\[
d(\Rvs,\Uvs A) \leq d(\Rvs,\Uvs A')
\]
\[\iff \dim(\Rvs \cap \Uvs A) \geq \dim(\Rvs \cap \Uvs A')\]
for all $A'\in G$. If $\Rvs$ is decodable then it holds that
\[
d(\Rvs,\Uvs A) < d(\Rvs,\Uvs A')
\]
for all $A'\in G\setminus \Stab_{G}(\Uvs A)$.





%
We will now explain how the canonizer mapping, explained in Section \ref{sec:orbitmetric}, can efficiently be computed for primitive cyclic orbit codes. 



\begin{lem}\label{lem:41}
Let $\Uvs,\Vvs\in \PP_q({n})$ be non-trivial. Then there exists $A\in G$ such that $\dim(\Uvs\cap \Vvs A)\ge 1$.
\end{lem}

\begin{proof}
Follows directly from the transitivity of $G$.
\end{proof}



\begin{lem}
Let $\Uvs,\Vvs\in  \PP_q({n})$ and $P^{i}\in G$ with $\dim(\Uvs\cap \Vvs P^{i})\ge 1$. Then there exist $u\in \Uvs$ and $v\in \Vvs$ such that $\phi(u)\phi(v)^{-1} \equiv \alpha^{i} \mod q^{n}-1$.
\end{lem}

\begin{proof}
Since $\dim(\Uvs\cap \Vvs P^{i})\ge 1$, there is  a non-zero vector $u\in \Uvs\cap \Vvs P^{i}$. Then there exists a vector $v\in \Vvs - \{0\}$ with $u=vP^{i}$. Hence, in field representation it holds that $\phi(u)\phi(v)^{-1} \equiv \alpha^{i} \mod q^{n}-1$.
\end{proof}

This leads to the following algorithm, for which we first need to define yet another vector space isomorphism:
\begin{align*}
\psi : \hspace{1.2cm} G &\longrightarrow \F_{q^{n}}\\
\sum_{i=0}^{n-1} \lambda_{i} P^{i} &\longmapsto \sum_{i=0}^{n-1} \lambda_{i} \alpha^{i}
\end{align*}
where $\lambda_{i} \in \F_{q}$.

\begin{algorithm}
\caption{Minimum distance decoder for primitive cyclic orbit codes in $\Gr$.}
\label{alg2}                       
\begin{algorithmic}       
\REQUIRE Code ${\cal C}=\Uvs \langle P\rangle \subseteq \Gr$, received word $
\Rvs \in \PP_q({n})$
\STATE set $d:=0, g:=0$
\FOR{$v \in \Rvs \setminus \{0\}$}
\FOR{$u \in \Uvs \setminus \{0\}$}
\STATE compute $g':=\phi(v)\phi(u)^{-1}$ in ${\mathbb F}_{q^n}$
\STATE compute $d':=\dim(\Rvs \cap \Uvs \psi^{-1}(g'))$
\IF{$d'>d$}
\STATE set $d:=d'$ and $g:=g'$
\ENDIF
\ENDFOR
\ENDFOR
\RETURN $\Uvs \psi^{-1}(g')$
\end{algorithmic}
\end{algorithm}

\begin{thm}
If $\dim(\Rvs):=k'$, the complexity of Algorithm \ref{alg2} is $\mathcal{O}(q^{k+k'}(n^2 + (k+k')^2 n))$ over $\F_q$. 
\end{thm}

\begin{proof}
\begin{enumerate}
 \item 
 In order to the determine the group element $g'$ one division in the finite field ${\mathbb F}_{q^n}$ must be performed which has complexity $\mathcal{O}(n^2)$.
\item
We compute the dimension of the intersection of subspaces 
with the aid of the Gaussian algorithm. Since $\dim(\Uvs)=k$ and $\dim(\Rvs)=k'$, the complexity is given by $\mathcal{O}((k+k')^2 n)$.
\item
The complexity of the algorithm is mainly determined by the two nested loops, i.\,e. the inner steps from above must be performed $(q^{\dim(\Uvs)}-1)(q^{\dim(\Rvs)}-1)$ times.
\end{enumerate}

\end{proof}

We can improve the complexity by choosing only specific elements of the received space to compute with in the algorithm. For this we need the following theorem.

\begin{thm}\label{t8}\cite{ma11j}
Let $v_{1},\dots,v_{k} \in \F_{q}^{n}$ be a basis of the sent code word $\Vvs\in \Gr$ and $r_{1},\dots, r_{k'} \in \F_{q}^{n}$ a basis of the received space $\Rvs$. Assume $f<k'$ linearly independent error vectors were inserted during transmission, i.e. $\Rvs= \Vvs'\oplus \mathcal{E}$, where $\Vvs'$ is a vector subspace of $\Vvs$ and $\mathcal{E}$ is the vector space spanned by the error vectors. Then the set
\[\mathcal{L}_{f} := \left\{\sum_{i\in I} \lambda_{i} r_{i} \mid  \lambda_{i} \in \F_{q} , I \in \binom{[k']}{f} \right\}\]
contains $k'-f$ linearly independent elements of $\Vvs$. By $\binom{[k']}{f}$ we denote the set of all subsets of $\{1,2,\dots,k'\}$ of size $f$.
\end{thm}

\begin{proof}
 Inductively on $f$: 
\begin{enumerate}
 \item 
If $f=0$, then $r_{1},\dots,r_{k'} \in\Uvs$.
 \item
If $f=1$, assume $r_1,\dots,r_{\ell}  \in \mathcal{E}$ and $r_{\ell+1},\dots,r_{k'}\in \Uvs'$. Then
\[r_i = \sum_{j=1}^k \lambda_{ij} v_j + \mu_i e \quad \forall \; i=1,\dots,l\]
where $e \in \mathcal{E}$ denotes the error vector, and hence $\forall \; i,h=1,\dots,l$
\[r_i + r_h = \sum_{j=1}^{k} (\lambda_{ij}+\lambda_{hj}) v_j + (\mu_i +\mu_h)e \]
\[\implies r_i + (-\mu_i \mu_h^{-1}) r_h = \sum_{j=1}^k (\lambda_{ij}+\lambda_{hj}) v_j .\]
Then the elements $r_{\ell+1},\dots,r_{k'}, r_1 + (-\mu_1 \mu_2^{-1}) r_2,\dots,$ $r_1 + (-\mu_1 \mu_{\ell}^{-1}) r_{\ell}$ are $k'-1$ linearly independent elements without errors.
 \item If more errors, say $e_{1},\dots,e_{f}$, were inserted, then one can inductively ``erase'' $f-1$ errors in the linear combinations of at most $f$ elements. Write the received elements as
 \[r_i = \sum_{j=1}^k \lambda_{ij} v_j + \sum_{j=1}^f \mu_{ij} e_{j} \quad \forall \; i=1,\dots,k' .\]
 Assume $\mu_{1f},\dots,\mu_{\ell f} \neq 0$ and $\mu_{(l+1)f},\dots,\mu_{k'f}=0$, i.e. the first $l$ elements involve $e_{f}$ and the others do not. 

 From above we know that the linear combinations of any two elements of $r_{1},\dots,r_{\ell}$ will include $l-1$ linearly independent elements without $e_{f}$. Denote them by $m_{1},\dots,m_{\ell -1}$. Naturally these elements are also linearly independent from $r_{\ell+1},\dots,r_{k'}$. Use the induction step on $m_{1},\dots,m_{\ell -1},r_{\ell+1},\dots,r_{k'}$ to get $k'-1-(f-1)=k'-f$ linearly independent elements without errors.
\end{enumerate}
\end{proof}



Hence, if the received word is decodable, the set $\mathcal{L}_{f}$ contains at least one element $v$ of the sent word $\Vvs$. For this element it holds that there exists $u \in \Uvs \setminus \{0\}$ such that
\[\dim(\Rvs \cap \Uvs \psi^{-1}(\phi(v)\phi(u)^{-1})) \geq k+k'-\delta +1\]
\[\iff d(\Rvs , \Uvs \psi^{-1}(\phi(v)\phi(u)^{-1})) \leq \delta -1 .\]
Therefore, Algorithm  \ref{alg2} is improved by replacing the outer FOR-loop ``\textbf{for} $v \in \Rvs \setminus \{0\}$ \textbf{do}'' by ``\textbf{for }$v \in \mathcal{L}_{f} \setminus \{0\}$ \textbf{do}''. 

If the number $f$ of inserted errors is unknown, one can set $f$ equal to the error-correction capability of the code. Since $d(\Rvs, \Vvs)\leq\delta-1$, the error-correction capability is equal to 
\begin{align*}
 \dim(\Evs) &= k' - \dim (\Rvs \cap \Vvs) \\
 &\leq k' -\frac{1}{2}(k'+k-\delta+1) \\
 &=\frac{1}{2}(k'-k+\delta-1) .
\end{align*}

More efficiently, the algorithm can iteratively work over $\mathcal{L}_{1}, \mathcal{L}_{2}, \dots , \mathcal{L}_{f}$ and return the respective result as soon as $d'\geq k+k'-\delta +1$ is found.   

The complexity of the improved algorithm differs from the one before only in the number of loops. Instead of $q^{k'}-1$ we now have  $\sum_{i=1}^{f+1} \binom{k'}{i} (q-1)^{i}$ elements of $\Rvs$ to check. 

Figures \ref{map3} and \ref{map5} compare the numbers $q^{k}-1$ (red graph) and $\sum_{i=1}^{f+1} \binom{k}{i} (q-1)^{i}$ (blue graph) for different values of $q$ and $f$ and thus depict that for increasing $q$ and small $f$ the complexity is greatly improved.

\begin{figure}[h]
\centering
\includegraphics[width=60mm]{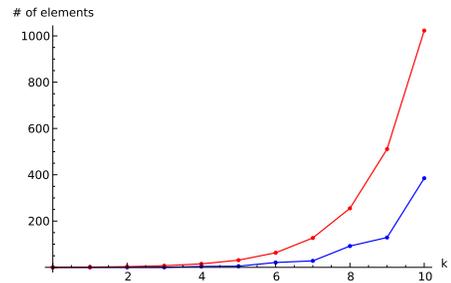}
\caption{$q=2, f=\lfloor k/2 \rfloor$-2}
\label{map3}
\end{figure}

\begin{figure}[h]
\centering
\includegraphics[width=60mm]{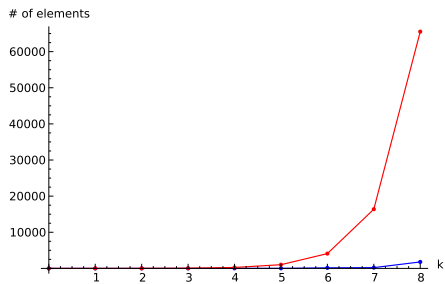}
\caption{$q=4, f=\lfloor k/2 \rfloor$-2}
\label{map5}
\end{figure}

In comparison, the decoding algorithms for Reed-Solomon like codes
contained in \cite{ko08} and \cite{si08a} in the case of $k=k'$ have a complexity of $\mathcal{O}(n^2(n-k)^2)$ and $\mathcal{O}(\delta (n-k)^3)$, respectively.  


In \cite{go11} the authors
present a minimum distance decoder for their spread code
construction. The complexity of their algorithm is
$\mathcal{O}((n-k)k^3)$. Another spread decoder is presented in \cite{ma11j}, whose complexity is $\mathcal{O}(nk(qk)^{f+1})$.


\begin{rem}
The presented algorithm also works for non-primitive irreducible cyclic orbit codes. The only main difference is that $G$ is not transitive and thus Lemma \ref{lem:41} does not hold for arbitrary elements. Nonetheless the subsequent results are still correct if one assumes that the received space $\Rvs$ is decodable. Moreover, in Lemma \ref{lem:41} ``$\alpha^{i}$'' has to be changed to an element from $\F_{q}[\alpha]$.
\end{rem}

\section{Conclusions}
\label{sec:conclusions}

In the first part we presented an overview of orbit codes in general and showed how these are a natural generalization of the concept of linearity for block codes. The main results of this part are that the minimum distance of the whole code is equal to the minimal distance between the initial point and any other point on the orbit and that one can define syndrome-like decoding for these codes.

Furthermore, we investigated cyclic orbit codes in the Grassmannian in more detail. For that we first showed how to classify them by the rational canonical forms of their generators. Then we showed how to compute the minimum distance of cyclic orbit codes and gave some examples of cardinalites and minimum distances found by random search. Moreover, we showed that spread codes can be constructed as primitive cyclic orbit codes for any set of valid parameters. In the end we explained how to decode irreducible cyclic orbit codes and determined the complexity of the proposed algorithm, which is efficient if the field size $q$ and the dimension of the vector spaces $k$ is small.

For further research it would be natural to generalize our presented results to arbitrary groups with more than one generator. Another interesting question is whether there are subgroups of the general linear group where the canonizing mapping is easily computed. For these groups an efficient minimum distance decoder can easily be derived. Furthermore, it is an open question what an encoder map from the actual message space could be. 

Although one loses some of the algebraic structure it is still an interesting project to investigate unions of orbit codes. In the primitive case some research in this area has already been done in \cite{ko08p}, but the more general case is unknown. Moreover, it is interesting to see what combinations of orbits will keep some algebraic structure and how this can be exploited for decoding.

\section*{Acknowledgement}

The authors thank Katherine Morrison for her useful comments on Section \ref{sec:cardinality}.

\bibliographystyle{plain}

\begin{thebibliography}{}

\end{thebibliography}


\begin{thebibliography}{10}

\bibitem{ah00}
R.~Ahlswede, N.~Cai, S.-Y.R. Li, and R.W. Yeung.
\newblock Network Information Flow.
\newblock {\em IEEE Transactions on Information Theory}, 46:1204--1216, 2000.

\bibitem{Art57}
E.~Artin.
\newblock {\em {Geometric Algebra}}.
\newblock Interscience Publishers, Inc., New York, 1957.

\bibitem{Bae52}
R.~Baer.
\newblock {\em {Linear Algebra and Projective Geometry}}.
\newblock Academic Press, New York, 1952.

\bibitem{Bir67}
G.~Birkhoff.
\newblock {\em {Lattice Theory}}.
\newblock American Mathematical Society, third edition, 1967.

\bibitem{Bra11}
M.~Braun.
\newblock {Lattices, Binary Codes, and Network Codes}.
\newblock {\em Advances in the Mathematics of Communications, Special Issue on Algebraic Combinatorics and Applications, April 11-18, 2010, Thurnau, Germany --- ALCOMA'10}, 5(2):225--232, 2011.

\bibitem{BEV11}
M.~Braun, T.~Etzion, and A.~Vardy.
\newblock {Linearity and Complements in Projective Space}.
\newblock {\em arXiv:1103.3117v1}, [cs.IT], 2011.

  
\bibitem{DM96}
J.~D. Dixon and B.~Mortimer.
\newblock {\em {Permutation Groups}}.
\newblock Springer-Verlag, 1996.

\bibitem{el10}
A.~Elsenhans, A.~Kohnert, and Alfred Wassermann.
\newblock Construction of Codes for Network Coding.
\newblock In {\em Proceedings of the 19th International Symposium on Mathematical Theory of Networks and Systems --- MTNS'10}, pages 1811--1814, Budapest, Hungary, 2010.

\bibitem{et08u}
T.~Etzion and N.~Silberstein.
\newblock Error-correcting codes in projective spaces via rank-metric codes and
  {F}errers diagrams.
\newblock {\em IEEE Transactions on Information Theory}, 55(7):2909--2919, 2009.
  
\bibitem{EV08} 
T.~Etzion and A.~Vardy. 
\newblock Coding Theory in Projective Spaces. 
\newblock In {\em Information Theory and Applications Workshop, 2008, Jan. 27th - Feb. 1st}, San Diego, USA, 2008.

\bibitem{et08p}
T.~Etzion and A.~Vardy.
\newblock Error-Correcting Codes in Projective Space.
\newblock In {\em IEEE International Symposium on Information Theory, 2008 --- ISIT 2008}, pages 871--875, 2008.

\bibitem{go11}
E.~Gorla, F.~Manganiello and J.~Rosenthal.
\newblock An Algebraic Approach for Decoding Spread Codes.
\newblock {\em arXiv:1107.55237v1}, [cs.IT], 2011.


\bibitem{he75}
I.~N. Herstein.
\newblock {\em Topics in algebra.} 2nd~ed.
\newblock Lexington, Mass.: Xerox College Publishing, 1975.

\bibitem{hi98}
J.~W.~P. Hirschfeld.
\newblock {\em Projective Geometries over Finite Fields}.
\newblock Oxford Mathematical Monographs. The Clarendon Press Oxford University Press, New York, second edition, 1998.
  
\bibitem{Ker99}
A.~Kerber.
\newblock {\em {Applied Finite Group Actions}}.
\newblock Springer-Verlag, 1999.

\bibitem{KSK09}
A.~Khaleghi, D.~Silva, and F.~R. Kschischang.
\newblock {Subspace Codes}.
\newblock In {\em Proceedings of the 12th IMA International Conference on Cryptography and Coding}, Cryptography and Coding '09, pages 1--21.
  Springer-Verlag, 2009.

\bibitem{ko08p}
A.~Kohnert and S.~Kurz.
\newblock Construction of Large Constant Dimension Codes with a Prescribed Minimum Distance.
\newblock In Jacques Calmet, Willi Geiselmann, and J\"orn M\"uller-Quade, editors, {\em MMICS}, volume 5393 of {\em Lecture Notes in Computer Science},
  pages 31--42. Springer, 2008.

\bibitem{ko08}
R.~K\"otter and F.~R. Kschischang.
\newblock Coding for Errors and Erasures in Random Network Coding.
\newblock {\em IEEE Transactions on Information Theory}, 54(8):3579--3591, 2008.
  
\bibitem{KM76}
E.~Kramer and D.~Mesner.
\newblock {$t$-Designs on Hypergraphs}.
\newblock {\em Discrete Mathematics}, 15(3):263--296, 1976.

\bibitem{li86}
R.~Lidl and H.~Niederreiter.
\newblock {\em Introduction to Finite Fields and their Applications}.
\newblock Cambridge University Press, Cambridge, London, 1986.

\bibitem{LW01}
J.~H. van Lint and R.~M. Wilson.
\newblock {\em {A Course in Combinatorics}}.
\newblock Cambridge University Press, second edition, 2001.

\bibitem{ma08p}
F.~Manganiello, E.~Gorla, and J.~Rosenthal.
\newblock Spread Codes and Spread Decoding in Network Coding.
\newblock In {\em Proceedings of the 2008 IEEE International Symposium on Information Theory (ISIT)}, pages 851--855, 6-11 July 2008.

\bibitem{ma11}
F.~Manganiello, A.-L. Trautmann, and J.~Rosenthal.
\newblock On Conjugacy Classes of Subgroups of the General Linear Group and Cyclic Orbit Codes.
\newblock In {\em Proceedings of the 2011 IEEE International Symposium on Information Theory (ISIT)}, pages 1916--1920, July 31 2011-Aug. 5 2011.

\bibitem{ma11j}
F.~Manganiello and A.-L. Trautmann.
\newblock Spread Decoding in Extension Fields.
\newblock {\em arXiv:1108.5881v1}, [cs.IT], 2011.

\bibitem{si08a}
D.~Silva, F. R. Kschischang, and R.~K\"otter.
\newblock A Rank-Metric Approach to Error Control in Random Network Coding.
\newblock {\em IEEE Transactions on Information Theory}, 54(9):3951--3967, 2008.
  
\bibitem{sk10}
V.~Skachek.
\newblock Recursive Code Construction for Random Networks.
\newblock {\em IEEE Transactions on Information Theory}, 56(3):1378--1382, 2010
  
\bibitem{Sle68}
D.~Slepian.
\newblock {Group Codes for the Gaussian Channel}.
\newblock {\em Bell System Technical Journal}, 47:575--602, 1968.
  
\bibitem{Tho87}
S.~Thomas.
\newblock {Designs over Finite Fields}.
\newblock {\em Geometriae Dedicata}, 24:237--242, 1987.

\bibitem{tr10p}
A.-L. Trautmann, F.~Manganiello, and J.~Rosenthal.
\newblock Orbit Codes---A new Concept in the Area of Network Coding.
\newblock In {\em IEEE Information Theory Workshop, Dublin, Ireland, August 2010 --- ITW 2010}, pages 1--4, 2010.

\bibitem{tr11}
A.-L. Trautmann and J.~Rosenthal.
\newblock A Complete Characterization of Irreducible Cyclic Orbit Codes.
\newblock In {\em Proceedings of the Seventh International Workshop on Coding and Cryptography --- WCC 2011}, pages 219--223, 2011.
  

\end{thebibliography}

\end{document}